\documentclass[letterpaper]{scrartcl}
\usepackage[utf8]{inputenc}
\usepackage[T1]{fontenc}
\usepackage[english]{babel}
\usepackage{amsmath}
\usepackage{amsthm}
\usepackage{amssymb}
\usepackage{color}
\usepackage{todonotes}
\usepackage{enumitem}
\usepackage{thmtools}
\usepackage{thm-restate}
\usepackage{authblk}

\usepackage{geometry}
\newtheorem{corollary}{Corollary}
\newtheorem{definition}{Definition}

\newtheoremstyle{example}
{0pt}
{2pt}
{\itshape}
{}
{\bfseries}
{.}
{.5em}
{}
\theoremstyle{example}

\title{\vspace{-5ex}The Primal-Dual Greedy Algorithm for Weighted Covering Problems}
\date{\vspace{-5ex}}

\author[1]{Britta Peis}
\affil[1]{RWTH Aachen University, \{peis,wierz\}@oms.rwth-aachen.de}
\author[2]{José Verschae}
\affil[2]{Pontificia Universidad Católica de Chile, jverschae@uc.cl}
\author[1]{Andreas Wierz}

\newcommand{\ZZ}{\mathbb{Z}}
\newcommand{\RR}{\mathbb{R}}

\newlist{enumerate-A}{enumerate}{1}
\setlist[enumerate-A]{label=A\arabic*}

\newlist{enumerate-prop}{enumerate}{1}
\setlist[enumerate-prop]{label=(Prop\arabic*)}

\newlist{enumerate-prop'}{enumerate}{1}
\setlist[enumerate-prop']{label=(Prop'\arabic*)}
\allowdisplaybreaks

\newcommand{\sectionheadline}[1]{\paragraph{#1.}}

\begin{document}
\maketitle

\begin{abstract}
We present a general approximation framework for weighted integer covering problems.
In a weighted integer covering problem, the goal is to determine a non-negative integer solution $x$ to  system $\{ Ax \geq r \}$ minimizing a non-negative cost function $c^T x$ (of appropriate dimensions).
All coefficients in matrix $A$ are assumed to be non-negative.
We analyze the performance of a very simple primal-dual greedy algorithm and discuss conditions of  system $(A,r)$ that guarantee feasibility of the constructed solutions, and a bounded approximation factor.

We call  system $(A,r)$ a \emph{greedy system} if it satisfies certain properties introduced in this work. These properties highly rely on monotonicity and supermodularity conditions on $A$ and $r$, and can thus be seen as a far reaching generalization of contra-polymatroids.
Given a greedy system $(A,r)$, we carefully construct a truncated system
$(A',r)$
containing  the same integer feasible points.
We show that our primal-dual greedy algorithm when applied to the truncated system $(A',r)$ obtains a feasible solution to $(A,r)$ with approximation factor at most $2\delta + 1$, or $2\delta$ if $r$ is non-negative.
Here, $\delta$ is some characteristic of the truncated matrix $A'$ which is small in many applications.
The analysis is shown to be tight up to constant factors.

We also provide an approximation factor of $k (\delta + 1)$ if the greedy algorithm is applied to the intersection of multiple greedy systems.
The parameter $k$ is always bounded by the number of greedy systems but may be much smaller.
Again, we show that the dependency on $k$ is tight.

We conclude this paper with an exposition of classical approximation results based on primal-dual algorithms that are covered by our framework.
We match all of the known results. Additionally, we provide some new insight in a generalization of the flow cover on a line problem.
\end{abstract}

\section{Introduction}
\label{sec:introduction}

Throughout this paper, we discuss integer covering problems of the form
\begin{align*} \min_{x \in \ZZ_+^n}  \left\{ c^T x \mid Ax \geq r \right\} \tag{P} \label{LP:P} \end{align*}
where $A \in \RR^{m \times n}_+$, $r \in \RR^m$, and $c\in \RR_+^n$.
We denote the row index set by $\mathcal{L}$ and the column index set by $E$.
The entries of matrix $A$ are denoted by $a_{i,e}$ for row $i \in \mathcal{L}$ and column / item $e \in E$, and can be interpreted as the weight of element $e$ with respect to row / constraint $i$.
We call the constraints $\{a_i^T x \geq r(i),\ i\in \mathcal{L}\}$ \emph{weighted covering constraints} since they ensure that the sum of weights of the multiplicities of items selected into a solution $x\in \ZZ_+^{|E|}$ cover $r(i)$, for each $i\in \mathcal{L}$.
Throughout, we call $r(i)$  the \emph{rank of $i$}, and denote the support of $ i \in \mathcal{L}$ by $\text{supp}(i) = \{ e \in E : a_{i,e} > 0 \}$.
Note that the non-negativity assumption on the cost vector $c \in \RR_+^{|E|}$ is non-restrictive, as variables with non-positive objective value may be selected infinitely often, thus, either rendering constraints redundant or rendering the problem instance unbounded.
For ease of notation, we assume that there is a trivial row $i \in \mathcal{L}$ with $supp(i) = \emptyset$ and $r(i) \leq 0$.
And without loss of generality, we assume that there is some row $i' \in \mathcal{L}$ with $supp(i') = E$.

Many interesting combinatorial optimization problems can be formulated in this manner: subset cover, cut covering, optimization over contra-polymatroids, and the knapsack cover problem, to mention just a few.
Consider, for example, the special case where $A$ is the incidence matrix of the family of all subsets $S$ of a finite set $E$.
That is, $\mathcal{L} = 2^E$ and, for each subset $S \subseteq E$, $a_{S,e} = 1$ if $e \in S$, and $a_{S,e} = 0$, otherwise.
If $r: 2^E \to \RR$ satisfies the three conditions (i) $r(\emptyset) = 0$, (ii) $r(S) \leq r(T)$ whenever $S \subset T$, and (iii) $r(S) + r(T) \leq r(S \cap T) + r(S \cup T)$ for all $S, T\subseteq E$, the polytope $\{x\in \RR_+^{|E|}\mid Ax\ge r\}$ corresponding to system $(A,r)$ is called \emph{contra-polymatroid}.
The well-known primal-dual \emph{(contra-) polymatroid greedy algorithm} \cite{edmonds1970submodular} determines for each cost function $c$ an optimal integral solution for $(\ref{LP:P})$ and its dual linear program
\begin{align*} \max_{y \in \RR_+^{|\mathcal{L}|}}  \left\{ y^T r \mid \sum_{i \in \mathcal{L}} a_{i,e} y_i \leq c_e \; \forall e \in E \right\}. \tag{D} \label{LP:D} \end{align*}
The polymatroid greedy algorithm will be described below. Conditions (ii) and  (iii) are usually called \emph{monotonicity} and \emph{supermodularity}, respectively.
The goal of this paper is to develop and analyze an extension of the primal-dual polymatroid greedy algorithm towards more general systems $(A,r)$ which may consist of an arbitrary integral matrix $A \in \RR_+^{|\mathcal{L}| \times |E|}$ and an arbitrary rank function $r:\mathcal{L} \to \RR$.
Our primal-dual algorithm will return primal and dual candidate solutions, $x$ and $y$ with the properties: $y$ is a feasible solution to the dual problem (D), and $x$ is an integral vector.
We establish conditions on system $(A,r)$ which guarantee 1) feasibility of the primal solution $x$ and 2) a bounded performance guarantee.
We distinguish a primal and dual phase of the algorithm.
During the dual phase, the algorithm constructs a feasible dual solution $y^*$ together with a collection of \emph{bottleneck elements} $E^* \subseteq E$ corresponding to a set of tight dual constraints with respect to $y^*$.
The primal phase  assigns non-negative integral values to the bottleneck elements in such a way that the primal constraints are fulfilled at least on the support of $y^*$.

\sectionheadline{Dual phase}
In general, given a feasible dual solution $y$, we call  $i \in \mathcal{L}$ \emph{augmentable} if there exists some positive amount $\epsilon > 0$ such that $y + \epsilon \chi_i$ remains a feasible solution (here, as usual, $\chi_i\in \{0,1\}^{|\mathcal{L}|}$ is all-zero, except for component $i$, which is $1$).
Recall the dual phase of the polymatroid greedy algorithm:
Starting with the all-zero vector $y\equiv 0$, the algorithm iteratively selects an augmentable $i \in \mathcal{L}$ of largest rank
(as long as augmentable variables exist), and raises $y_i$ as far as possible, that is, until the dual constraint of some element $e \in E \setminus E^*$ becomes tight.
If we define the support of $i \in \mathcal{L}$ by $S_i = supp(a_i)$, the algorithm always selects an augmentable $i \in \mathcal{L}$ with inclusion-wise maximal set $S_i$.
A similar dual greedy approach can be applied to general packing problems of type (\ref{LP:D}) with arbitrary matrix $A \in \RR_+^{|\mathcal{L}| \times |E|}$.
In fact, it is probably the most naive approach one can think of:
Take some precedence rule $(\mathcal{L}, \preceq)$ on the row index set $\mathcal{L}$ such that the rank is monotone with respect to the precedence rule, that is, $i \preceq j$ implies $r(i) \leq r(j)$.
Apply the following iterative procedure.
\begin{enumerate}
\item Initially, let $y^* \equiv 0$.
\item While $\mathcal{L} \neq \emptyset$
\begin{enumerate}
\item Select $i \in \mathcal{L}$ with $r(i)$ maximum. If there are multiple choices, select one which is maximum with respect to $(\mathcal{L}, \preceq)$. \label{alg:dual-phase:select-row}
\item STOP if $r(i) \leq 0$. \label{alg:dual-phase:stop}
\item Raise $y^*_i$ until some element $e^* \in E \setminus E^*$ becomes tight.
\item Add $e^*$ to $E^*$ and iterate with $\mathcal{L} = \{i \in \mathcal{L} : a_{i, e^*} = 0\}$. \label{alg:dual-phase:lattice-subtract}
\end{enumerate}
\end{enumerate}
Certainly, this approach always returns a feasible dual solution $y^*$.
The performance of this algorithm can, however, be arbitrarily bad, even in case of binary matrices.
In order to get upper bounds on the performance guarantee of the algorithm, let us first extend the dual phase by an associated primal phase.

\sectionheadline{Primal phase}
Note that, for each problem of type (D), the dual greedy algorithm returns a feasible dual solution $y^* \in \RR_+^{|\mathcal{L}|}$ whose support forms a sequence
$$\text{supp}(y^*)=\{i_1, \dots, i_{\ell+1}\}$$ in the order in which variables were considered during the dual phase.
Here, we assume that $i_{\ell+1}$ is actually not raised but used in the STOP criterion \ref{alg:dual-phase:stop}).
Recall that we assumed that there is always some trivial row in $A$ with $r(i) \leq 0$ and $supp(i) = \emptyset$.
Thus, the STOP criterion is reached after at most $|E|$ iterations.
Moreover, the choice of variables to be increased during the dual phase implies $r(i_1) \geq \dots \geq r(i_\ell) > 0 \geq r(i_{\ell+1})$.
Let $E^*=\{e_1, \dots, e_\ell\} \subseteq E$ be the associated bottleneck-elements satisfying
$$e_j\in \text{supp}(a_{i_j}) \quad \mbox{ and } \quad a_{i_k, e_j}=0 \quad \forall 1 \leq j < k \leq \ell + 1.$$

In the special case where $(A,r)$ describes a contra-polymatroid, the sequence $\{i_1, \dots, i_{\ell+1}\}$ corresponds to a chain of sets $E_{\ell+1} \subset \dots \subset E_1$
with $E_1 = E$ and $E_{j + 1} = E_j \setminus {\{e_j\}}$ for $j = 1, \dots, \ell$.
Moreover, the chain satisfies $r(E_\ell) > 0 = r(E_{\ell+1})$.
The primal phase of the polymatroid greedy algorithm simply constructs a primal vector $x \in \ZZ_+^{|E|}$ by setting $x_{e_\ell} = r(E_\ell)$ and $x_{e_j} = r(E_j) - r(E_{j+1})$ for $j = \ell-1$ down to $j = 1$.

A natural extension of the polymatroid greedy algorithm towards more general systems $(A,r)$ is the following procedure:
Given the sequence $\{i_1, \dots, i_{\ell+1}\} \subseteq \mathcal{L}$ and $E^* = \{e_1, \dots, e_\ell\} \subseteq E$ as constructed during the dual phase, set
\begin{align}
x^*_{e_j} = \left\lceil \frac{r(i_j)^+ - r(i_{j+1})^+}{a_{i_j,e_j}} \right\rceil \quad j = 1, \dots, \ell. \label{alg:primal-phase:set-x}
\end{align}
Here, $r(i)^+ = \max\{r(i), 0\}$ is the positive part of the rank.
Our primal-dual greedy algorithm for weighted covering and packing problems of type (\ref{LP:P}) and (\ref{LP:D}) consists of the concatenation of the dual and primal phase as described above.
In this paper, we discuss properties of system $(A,r)$ that ensure the following two properties of the primal-dual greedy algorithm:
 (1) the primal greedy solution $x^*$ is feasible and
 (2) has a bounded performance guarantee.
We also provide complementing complexity theoretical results and lower bounds on integrality gaps under the assumed properties.

\sectionheadline{Related work}

Integrality of polyhedra described by systems $(A,r)$ plays an important role in combinatorial optimization.
Probably one of the most-famous conditions of a system $(A,r)$ in order to guarantee integrality of the polyhedron is totally-unimodularity of matrix $A$ and integrality of $r$, which was discussed by Hoffmann \cite{hoffman1976total}.
The same effect appears if $(A,r)$ is totally dual integral, a condition introduced by Giles and Pulleybank \cite{giles1979total}.

Instead of solving a linear program to optimality via a general purpose linear programming algorithm, the primal-dual method was widely used in order to obtain optimal solutions.
Many classical algorithms can also be cast as primal-dual methods, see e.g.\ Williamson and Shmoys \cite{williamson2011design} or Papadimitriou and Steiglitz \cite{papadimitriou1982combinatorial} for an overview of such connections.
For matrices $A$ with coefficients in $\{-1,0,1\}$, many structural properties ensuring optimality of the primal-dual method were studied.
Optimization over polymatroids is probably one of the most famous results of this type due to Edmonds \cite{edmonds1970submodular}.
Following this result, lots of generalizations were established, such as \cite{faigle2009general,faigle2000order,faiglegreedy,faigle2012ranking,faigle2010two,frank1999increasing,fujishige1984note,fujishige2004dual}, to mention just a few.
Sub- or supermodularity of $A$ (and/or $r$) usually plays an important role in these optimality results.

Almost four decades ago, the primal-dual method was first used in order to obtain approximation algorithms for integer programs.
Bar-Yehuda and Even obtained a primal-dual approximation algorithm for vertex cover \cite{bar1981linear}.
Following this work, network design problems were considered e.g.\ by Agrawal et al. \cite{agrawal1995trees} or Goemans and Williamson \cite{goemans1995general}.
The latter introduced a fairly general framework which applies to lots of network design problems modelled via so-called proper functions.
See e.g.\ Bertsimas and Tao \cite{bertsimas1998valid}, Williamson and Shmoys \cite{williamson2011design} or Vazirani \cite{vazirani2013approximation} for surveys.

Still, most approximation results in this direction consider only matrices $A$ with coefficients in $\{-1,0,1\}$.
Carnes and Shmoys \cite{carnes2008primal} considered the knapsack cover problem and showed that a formulation with strengthened inequalities can be solved via the primal-dual method.
The strengthening is due to Carr et al. \cite{carr2000strengthening}.
Bar-Noy et al.\ \cite{bar2001unified} solved the flow cover on a line problem via a local-ratio technique.
This technique can equivalently be seen as a primal-dual approach.

\sectionheadline{Our contribution and structure of the paper}

\begin{table}[tb]
 \begin{center}
 \begin{tabular}{l|l l}
          & \multicolumn{2}{c}{Approximation factor} \\
  Problem & Best known & Our bound \\ \hline
  Optimization over contra-polymatroids & 1 \cite{edmonds1970submodular} & 1 $^*$ \\
  Knapsack cover & 2 \cite{carnes2008primal,mccormick2016primal} & 2 $^*$ \\
  Subset cover & $\log (\max_i |T_i|)$ \cite{chvatal1979greedy} & $\max_i |T_i|$ $^*$ \\
  $p$-Contra-polymatroid intersection & $p$ \cite{jenkyns1976efficacy} & $p$ $^\dagger$ \\
  Flow cover on $k$ lines & $k = 1: 4$ \cite{bar2001unified} & $4k$ $^\dagger$ \\
                          & $k > 1:$ none & \\
  Knapsack cover with precedence constraints & $w$ \cite{mccormick2016primal} & $w$ $^\dagger$ \\
  Generalized steiner tree & $2$ \cite{goemans1997primal} & $2$ $^\dagger$ \\
  Minimum multicut on trees & $2$ \cite{garg1997primal} & $2$ $^\dagger$ \\
 \end{tabular}
 \caption{Exemplary results derivable from this work. $^*$ via Theorem \ref{thm:sv:greedy-approximation}, $^\dagger$ via Theorem \ref{thm:ps:greedy-approximation}.}
 \label{table:intro:result-summary}
 \end{center}
\end{table}

We call a system $(A,r)$ a \emph{greedy system}, if it satisfies certain properties which are formalized in Section \ref{sec:simple-version}.
Intuitively, these properties can be seen as generalizations of properties that define matroids.
In contrast to matroids, however, our results do not necessarily provide optimum solutions but only approximation guarantees.
Despite the fact that greedy systems can have an unbounded integrality gap, we can show that a careful truncation of coefficients of matrix $A$ obtains strong approximation results.
We provide an approximation factor of $(2\delta + 1)$ or $2 \delta$, if $r$ is non-negative.
Under certain conditions, the additional factor of $2$ vanishes.
The characteristic $\delta$ depends on the range of coefficients in the truncated matrix, which is small in many applications.
We also show that there are greedy systems such that the ratio between an optimum dual solution and the dual solution constructed by the greedy algorithm is $\delta$.
This implies that the dependency on $\delta$ in our approximation ratio can not be improved (up to constant factors) unless the type of solution computed in the dual phase is changed.
Finally, we show that the properties required for a greedy system are necessary in order to ensure that the discussed greedy algorithm always obtains a feasible solution.

To be able to further increase our modelling power, we provide a generalization in Section \ref{sec:product-version}.
The generalization can be seen as a composition of system $(A,r)$ of multiple greedy systems on the same column set.
We call such a system a greedy product system.
For greedy product systems, we can obtain similar results, proving approximation guarantees of $k (\delta + 1)$, or $k \delta$ if the truncation is a binary matrix.
Here, $k$ is a characteristic that is problem specific, but small in the discussed applications.
Again, we can show that the dependency on $k$ can not be improved unless the type of solution constructed in the dual phase is changed.

Table \ref{table:intro:result-summary} contains some exemplary results derivable from this paper.
A detailed discussion of the table including a description of the problems and all proofs can be found in the appendix.
Although the result regarding subset cover does not coincide with the best-known result, this discrepancy is possibly expected.
To the best of our knowledge, no logarithmic approximation guarantee based on a primal-dual analysis for subset cover is known.
Instead, primal averaging arguments are commonly used.
The result for flow cover on $k$ lines was not known before and looks like a promising direction of further modelling applications.

\section{Sufficient conditions for feasibility and bounded performance}
\label{sec:simple-version}

In this section, we will discuss sufficient conditions for system $(A,r)$ in order to ensure that the primal solution obtained by the primal-dual greedy algorithm (from the previous section) is always feasible and has a bounded performance guarantee.
Throughout this section, we assume that some given partial order $(\mathcal{L},\preceq)$ is fixed which is used in order to choose a dual variable to be increased.

We call a system $(A,r)$ a \emph{greedy system} (with respect to $(\mathcal{L},\preceq)$), if it satisfies the following properties.
Whenever we talk about a greedy system in the remainder of this work, we always assume that this is with respect to this fixed partial order.
In order to simplify notation, we will use $S \in \mathcal{L}$ to denote the row $i$ with $supp(i) = S \subseteq E$.
One of the subsequent properties will ask for the support of rows to be unique, hence, we can talk about \emph{the} rows.
\begin{enumerate-prop}
 \item $r$ is monotone non-decreasing on $(\mathcal{L}, \preceq)$: $r(S) \leq r(T)$ for all  $S \preceq T$. \label{prop:first} \label{prop:r-monotone}
 \item For each element $e \in E$, $a_{*,e}$ is monotone non-decreasing on $(\mathcal{L}, \preceq)$: $a_{S,e} \leq a_{T,e}$ for all $S \preceq T$. \label{prop:A-monotone}
 \item $(\mathcal{L},\preceq)$ is a modular lattice with join $\vee$ and meet $\wedge$ such that $i,j \in \mathcal{L}$ with $i \neq j$ implies $supp(i) \neq supp(j)$.
 Moreover, we require that for all $i,j \in \mathcal{L}$ and $e \in E$ it is true that $e \not\in supp(i) \cup supp(j) \Rightarrow e \not\in supp(i \vee j)$. \label{prop:lattice-modular}
 \item The system $(A,r)$ is \emph{weighted supermodular} on $(\mathcal{L},\preceq)$: \label{prop:A-r-coupling}
 $$ \frac{r(T) - r(S \wedge T)}{a_{T,e}} \leq \frac{r(S \vee T) - r(S)}{a_{S \vee T,e}} \quad \forall S,T \in \mathcal{L}, e \in T \setminus (S \wedge T).$$ \label{prop:last}
\end{enumerate-prop} \vspace*{-0.5cm}

In this paper, we will often talk about matrix $A$ being monotone.
By this, we mean that \ref{prop:A-monotone} holds.
A partial order $(\mathcal{L}, \preceq)$ is called a lattice if for any two elements $i,j \in \mathcal{L}$ there is a unique least common upper bound (join $i \vee j = \inf\{k \in \mathcal{L} : i,j \preceq k \}$) and a unique greatest common lower bound (meet $i \wedge j = \sup\{ k \in \mathcal{L} : k \preceq i,j \}$).
In case of the Boolean lattice $(2^E, \subseteq)$, the join and meet are set union and intersection, respectively.
A lattice is called modular, if for all $i,j,k \in \mathcal{L}$ with $i \preceq k$ the following holds: $i \vee (j \wedge k) = (i \vee j) \wedge k$.
A subset $I \subseteq \mathcal{L}$ is called a \emph{chain} if the ordering relation $\preceq$ yields a total order of $I$.
A chain is called \emph{dense}, if there is no $k \in \mathcal{L} \setminus I$ which can be added to $I$ without violating the chain property.
Modularity will be an important property as it implies the following:
Let $i,j \in \mathcal{L}$ and consider any dense chains $I_i \subseteq \{k \in \mathcal{L} : i \preceq k \preceq i \vee j \}$ and $I_j \subseteq \{ k \in \mathcal{L} : i \wedge j \preceq k \preceq j \}$.
Modularity implies that there is an isomorphism $\psi: I_i \rightarrow I_j$ (c.f.\ Theorem 13 in \cite{birkhoff1948lattice}).
The fact that this isomorphism exists will help in order to prove feasibility and a bounded approximation factor of our constructed primal solution.
For more information on lattice theory, see e.g.\ \cite{birkhoff1948lattice}.

Note that \ref{prop:A-monotone} and the fact that the support of row $S \in \mathcal{L}$ equals $S$ implies that the support on chains is monotonically increasing, that is, $S \preceq T$ implies $S \subseteq T$.
For a row $S \in \mathcal{L}$ and element $e \in E$, we use the notation $\phi_e(S) = \max \{ T \in \mathcal{L} \mid T \preceq S, e \not\in T \}$ to denote the maximum row $T \preceq S$, which does not contain element $e$ in its support.
Observe that $\phi$ always returns a unique element due to Lemma \ref{lem:sv:sublattice}.
The removal of an element from the lattice obtains a sublattice of $\mathcal{L}$.
Note that $\phi_e(S) = S$, if $e \not\in S$.
The function $\phi_e(S)$ has a strong connection with Line~\ref{alg:dual-phase:lattice-subtract}) in the dual phase of the greedy algorithm:
If $S_{\ell+1} \prec \dots \prec S_1$ is the support of the dual solution obtained with bottleneck elements $e_i, 1 \leq i \leq \ell$, then $S_{i+1} = \phi_{e_i}(S_i)$.
In order to make this observation, it is helpful to realize that the restriction in Line~\ref{alg:dual-phase:lattice-subtract}) describes a sublattice of $\mathcal{L}$.
Since $\mathcal{L}$ is a lattice, it always contains a unique maximum element (the join of all elements).
Hence, also the row chosen in Line~\ref{alg:dual-phase:select-row}) will always be unique.
And due to \ref{prop:r-monotone}, the maximum row will also have the maximum rank value.
In other words, instead of computing the sublattice in iteration $i$ of Line~\ref{alg:dual-phase:lattice-subtract}) explicitly, we can select $S_{i+1} = \phi_{e_i}(S_i)$ in order to choose the row to be considered in the subsequent iteration.

We will use two observations in the following sections.
These will ensure that the step in Line \ref{alg:dual-phase:lattice-subtract}) maintains the properties of a greedy system.

\begin{restatable}{lemma}{svObsPhiOrderPreserving}
 Let $S,T \in \mathcal{L}$ and $e \in E$, then $\phi_e(S) \preceq \phi_e(T)$.
\end{restatable}

\begin{proof}
 See appendix for the proof.
\end{proof}

\begin{restatable}{lemma}{svObsSublattice}
 Let $(A,r)$ be a greedy system and $e \in E$.
 Then the system restricted to $\mathcal{L}' = \{ L \in \mathcal{L} : e \not\in L \}$ is a greedy system.
 \label{lem:sv:sublattice}
\end{restatable}

\begin{proof}
 See appendix for the proof.
\end{proof}

\sectionheadline{Feasibility}
Let us assume that the greedy algorithm returns the dual solution $y^*$ with support (in the order the sets occurred during the dual phase) $S_1, \dots, S_{\ell+1}$ and bottleneck elements $e_1,\dots,e_\ell$, and let $x^*$ be the corresponding primal vector.
Then $r(S_{\ell+1}) \leq 0 < r(S_\ell)$.

First, note that \ref{prop:r-monotone} and \ref{prop:A-monotone} and the choice of $S$ in every iteration ensures that $S_{\ell+1} \prec S_\ell \prec \dots \prec  S_1$ forms a chain in $(\mathcal{L},\preceq)$.
Moreover, the rank differences considered in order to construct $x^*$ are always non-negative.
Hence, $x^*$ will be a non-negative vector.
The following Lemma \ref{lem:sv:set-feasible} shows that it will also be a feasible primal solution.

Before we prove feasibility of $x^*$, we will obtain two simple observations.
The first is regarding the marginal increase version of supermodularity, the second is regarding the behavior of elements just before the rank becomes negative.

\begin{restatable}{lemma}{obsSvMarginalIncreaseSupermodularity}
 Let $S \preceq T \in \mathcal{L}$ and let $e \in S$.
 Then $\frac{r(S) - r(\phi_e(S))}{a_{S,e}} \leq \frac{r(T) - r(\phi_e(T))}{a_{T,e}}$.
 \label{lem:sv:marginal-increase-supermodularity}
\end{restatable}

\begin{proof}
 See appendix for the proof.
\end{proof}

\begin{restatable}{lemma}{obsSvMonotoneAroundZero}
 Let $S \preceq T \in \mathcal{L}$ and let $e \in S$ such that $r(S), r(T) \geq 0$ and $r(\phi_e(S)), r(\phi_e(T)) \leq 0$.
 Then $\frac{r(S)}{a_{S,e}} \leq \frac{r(T)}{a_{T,e}}$
 \label{lem:sv:monotone-around-zero}
\end{restatable}

\begin{proof}
 See appendix for the proof.
\end{proof}

\begin{restatable}{lemma}{lemSvFeasible}
 Let $(A,r)$ be a greedy system and let $(x^*,y^*)$ be obtained by the primal-dual greedy algorithm.
 Then $x^*$ is feasible for $(\ref{LP:P})$.
 \label{lem:sv:set-feasible}
\end{restatable}

\begin{proof}
 Let $S_{\ell + 1} \prec \dots \prec S_1$ be the support of $y^*$.
 Again, we assume that $S_{\ell+1}$ was used during the STOP criterion, that is, $r(S_{\ell+1}) \leq 0 < r(S_{\ell})$.
 Let $T \in \mathcal{L}$ be any row with $r(T) > 0$.
 We will show that $a_T x \geq r(T)$ holds, where $a_T$ denotes the row of $A$ induced by index $T$.
 Let $e_1,\dots,e_\ell$ be the bottleneck elements and let us consider the following chain:
 $T'_1 = T, T'_{i+1} = \phi_{e_k}(T'_i)$ for $i = 1,\dots,\ell$.
 Then $r(T'_{\ell+1}) \leq 0$, as $T'_{\ell+1} \preceq S_{\ell+1}$.
 Moreover, for all $1 \leq i \leq \ell$, $T'_i \preceq S_i$ holds due to modularity of the lattice.

 Note that the chain $T'_i$ may contain the same element multiple times, that is, $T'_i = T'_{i+1}$ may hold for some $i$.
 In this case, $e_{i+1} \not\in T'_i$.
 Let $\alpha = \max\{1 \leq i \leq \ell : r(T'_i) > 0 \}$ be the maximum index such that $T'_\alpha$ has positive rank.
 An important observation is that $a_{T'_\alpha,e_\alpha} > 0$.
 If this was not the case, $e_\alpha \not\in T'_\alpha$, hence, $\phi_{e_\alpha}(T'_\alpha) = T'_\alpha$, that is, $\alpha$ could be increased.
 Moreover, let $I = \{1 \leq i < \alpha : a_{T'_i,e_i} > 0\}$ be the index set of all distinct chain elements, prior to element $T'_\alpha$.

 Then
 \begin{align*}
  a_T x
  &= \sum_{i=1}^{\ell} a_{T,e_i} x^*_{e_i}
  \overset{\ref{prop:A-monotone}}{\geq} \sum_{i=1}^{\ell} a_{T'_i,e_i} x^*_{e_i}
  \geq \sum_{i \in I} a_{T'_i,e_i} \frac{r(S_i)^+ - r(S_{i+1})^+}{a_{S_i,e_i}} + a_{T'_\alpha,e_\alpha} x^*_{e_\alpha} \\
  &\overset{Lem. \ref{lem:sv:marginal-increase-supermodularity}}{\geq} \sum_{i \in I} a_{T'_i,e_i} \frac{r(T'_i) - r(T'_{i+1})}{a_{T'_i,e_i}} + a_{T'_\alpha,e_\alpha} x^*_{e_\alpha}
  = r(T'_1) - r(T'_{\alpha}) + a_{T'_\alpha,e_\alpha} x^*_{e_\alpha}
 \end{align*}
 The second inequality uses the definition of $x^*$ and the definition of index set $I$.
 In the third inequality, we use the fact that only $r(S_{\ell+1})$ is possibly negative, hence, all coefficients in the sum were positive.

 The proof is almost concluded.
 If $\alpha = \ell$, we have
 \begin{align*}
  a_{T'_\alpha,e_\alpha} x^*_{e_\alpha}
  \geq a_{T'_\ell,e_\ell} \frac{r(S_\ell)}{a_{S_\ell,e_\ell}}
  \overset{Lem. \ref{lem:sv:monotone-around-zero}}{\geq} a_{T'_\ell,e_\alpha} \frac{r(T'_\ell)}{a_{T'_\ell,e_\ell}} = r(T'_\ell) = r(T'_\alpha),
  \end{align*}
  which implies $a_T x \geq r(T'_1) = r(T)$.
  Note that Lemma \ref{lem:sv:monotone-around-zero} is applicable, as $r(S_{\ell+1}) \leq 0$ and, moreover, $\phi_{e_\ell}(T'_\ell) \preceq S_{\ell+1}$, which implies by \ref{prop:r-monotone}, that $r(\phi_{e_\ell}(T'_\ell)) \leq r(S_{\ell+1}) \leq 0$.

  If $\alpha < \ell$, then
  \begin{align*}
    a_{T'_\alpha,e_\alpha} x^*_\alpha
    \geq a_{T'_\alpha,e_\alpha} \frac{r(S_\alpha) - r(S_{\alpha+1})}{a_{S_\alpha, e_\alpha}}
    \overset{Lem. \ref{lem:sv:marginal-increase-supermodularity}}{\geq} a_{T'_\alpha,e_\alpha} \frac{r(T'_\alpha) - r(\phi_{e_\alpha}(T'_\alpha))}{a_{T'_\alpha, e_\alpha}} = r(T'_\alpha) - r(\phi_{e_\alpha}(T'_\alpha)),
  \end{align*}
  which concludes the proof as $r(\phi_{e_\alpha}(T'_\alpha)) \leq 0$, by choice of $\alpha$.
\end{proof}

Unfortunately, formulations satisfying \ref{prop:first} - \ref{prop:last} may have an unbounded integrality gap.
The integrality gap of an integer program is the ratio between the value of an optimum fractional solution and the value of an optimum integral solution.
Consider for example the knapsack cover instance
$\min \{ D x_1 + x_2 \mid D x_1 + (D-1) x_2 \geq D, x \in \{0,1\}^2 \}.$
It can be reformulated in such a way that the explicit variable upper bounds are no longer required.
In its reformulated version, the system satisfies \ref{prop:first} - \ref{prop:last} and has an integrality gap of $D$.

In order to obtain a stronger LP relaxation of (\ref{LP:P}), we will truncate the coefficients of matrix $A$.
This will yield a stronger relaxation without cutting off any integer feasible solutions.
\begin{restatable}{definition}{defSvTruncation}
Let $(A,r)$ be a greedy system and define $A' \in \ZZ_+^{|\mathcal{L}| \times |E|}$ with coefficients as follows.
For $S \in \mathcal{L}$ and $e \in E$, set $a'_{S,e} = \min\{ a_{S,e}, r(S)^+ - r(\phi_e(S))^+ \}$.
We call the system $(A',r)$ the \emph{truncation} of $(A,r)$.
\label{def:sv:truncation}
\end{restatable}

The truncation ensures that the coefficient $a'_{S,e}$ of an element $e$ with respect to row $S$ reflects at most the difference between the positive parts of ranks $r(S)$ and $r(\phi_e(S))$.
The row $\phi_e(S)$ is chosen in such a way that it is the first row in which the coefficient of $e$ will vanish with respect to row $S$.
Hence, intuitively, variables $x_f, f \in E$ with a positive coefficient $a'_{\phi_e(S),f}$ have to ensure that the residual rank $r(\phi_e(S))$ is covered.
Hence, larger coefficients $a'_{S,e}$ should not help.

We will now prove that the truncated system
\begin{align*}
 \min_{x \in \ZZ^{|E|}_+}  \{ c^T x \mid A'x \geq r \} \tag{T} \label{LP:T}
\end{align*}
contains the same integer feasible points as the original system (\ref{LP:P}).

\begin{restatable}{lemma}{lemSvTruncationFeasible}
 Let $(A,r)$ be a greedy system with truncation $(A',r)$ and let $x \in \ZZ_+^{|E|}$.
 Then $x \in (\ref{LP:P})$ if and only if $x \in (\ref{LP:T})$.
 \label{lem:sv:truncation-feasible}
\end{restatable}

\begin{proof}
 See appendix for the proof.
\end{proof}

The truncated system (\ref{LP:T}) no longer necessarily satisfies \ref{prop:A-r-coupling}.
However, $A'$ will still be monotone and both polyhedra describe the same integer points.
Hence, one might ask if we can still apply the greedy algorithm to (\ref{LP:T}) in order to obtain a feasible solution.
And this is indeed true as shown in the following Lemma.

\begin{restatable}{lemma}{lemRfGreedyTruncationFeasible}
 The greedy algorithm applied to the truncation (\ref{LP:T}) of a greedy system $(A,r)$ obtains a feasible primal solution to (\ref{LP:T}) and (\ref{LP:P}).
 \label{lem:sv:greedy-truncation-feasible}
\end{restatable}

\begin{proof}
 See appendix for the proof.
\end{proof}

Note that, in case of the knapsack cover problem, the truncation coincides with the polyhedron discussed in \cite{carnes2008primal}, which has an integrality gap of at most 2.
One might hope that the integrality gap of the truncation is always bounded by a small constant.
But let us proceed with some negative results before we derive an approximation guarantee.

\sectionheadline{Inapproximability and integrality gaps}
We can construct simple examples with an integrality gap linear in the number of elements.
Moreover, we can show that one can not expect a $(1- o(1)) \log n$ approximation unless $NP \subseteq DTIME(n^{O(\log\log n)})$.

\begin{restatable}{proposition}{prpSvLogInapproximable}
The subset cover problem can be modeled in the form (\ref{LP:P}) as a greedy system $(A,r)$.
Hence, no $(1-o(1)) \log n$ approximation for (\ref{LP:T}) exists unless $NP \subseteq DTIME(n^{O(\log\log n)})$.
\label{prp:sv:log-inapproximable}
\end{restatable}

\begin{proof}
See appendix for the proof.
\end{proof}

\begin{restatable}{proposition}{prpSvLinearGap}
There exists a family of instances such that the truncation (\ref{LP:T}) has an integrality gap linear in the number of elements.
\label{prp:sv:linear-gap}
\end{restatable}

\begin{proof}
See appendix for the proof.
\end{proof}

\sectionheadline{Approximation guarantee}

The preceding two results might suggest that the truncation does not help in order to obtain general approximation guarantees.
But a deeper look into the required properties of $(A,r)$ reveals a better understanding.
In both constructions, the efficiency of elements dropped drastically before dropping to zero.
The instance considered in Proposition \ref{prp:sv:linear-gap}, is very simple and defined on the Boolean lattice $\mathcal{L} = 2^E$ with $|E| = n$.
Each element $e \in E$ has the same weight $a_{S,e} = 2^{|E|}$ for all $S \subseteq E, e \in E$ and the rank function is symmetric and defined as $r(S) = 2^n(2^{-(n-|S|)} - 2^{-\frac{n}{2}})$.
That is, $r(S)$ for $S \in \mathcal{L}$ is positive if and only if $|S| > \frac{n}{2}$.
Moreover, the marginal differences $r(S) - r(\phi_e(S))$ are exponentially decreasing with decreasing cardinality of set $S$.
This implies that the truncation has large coefficients $a'_{E,e}$ and very small coefficients $a'_{S,e}$ for $S \in \mathcal{L}$ with $|S| = \frac{n}{2} + 1$ (assuming $n$ even).
In particular, $a'_{E,e} \gg a'_{S,e}$ holds for these sets, which will pose as an issue for small approximation factors.
A similar effect may appear if the ratio between coefficients in the matrix $A$ is large, that is, if $a_{E,e} \gg a_{S,e} > 0$ holds for some $S \in \mathcal{L}$ and $e \in E$.

Let us consider the following parameter as a measure of the range of efficiency of elements.
Given $S \in \mathcal{L}$ and $e \in E$, define
\begin{align*}
 \delta_{S,e} = \begin{cases} \frac{a'_{E,e}}{a'_{S,e}} & a'_{S,e} > 0 \text{ and } (r(\phi_e(S)) \geq 0 \text{ or } a'_{S,e} = a_{S,e}), \\ 1 & \text{otherwise}, \end{cases}
\end{align*}
and let $\delta = \max \{ \delta_{S,e} : S \in \mathcal{L}, e \in E \}$.
Then $\delta$ can be used in order to bound the ratio between coefficients in matrix $A'$.
Note that we excluded certain rows and elements from the bound in order to possibly make it smaller.
If the rank is non-negative, this does not have any impact.
However, if the rank is possibly negative, we get better bounds as we will see later.

And indeed, we can estimate the quality of a solution obtained by the greedy algorithm in terms of $\delta$ as the following theorem shows.
The role of $b$ in the theorem can be seen as follows.
Recall how the primal phase constructs the vector $x^*$ in (\ref{alg:primal-phase:set-x}).
For $b = 1$, the rounding does not affect $x^*_e$ for any bottleneck element $e \in E$, that is, the vector will naturally be integral and, in particular, $a'_{S_i,e} x^*_e = r(S_i)^+ - r(S_{i+1})^+$ will hold for all $1 \leq i \leq \ell$.
If $b = 2$, the rounding may have an impact on $x^*_e$.
In this case, it is easy to see that the marginal rank differences are oversubscribed by at most a factor of two, that is, $a'_{S_i,e} x^*_e \leq 2 (r(S_i)^+ - r(S_{i+1})^+)$.

\begin{restatable}{theorem}{thmSvGreedyApproximation}
 Let $(x^*,y^*)$ be a solution returned by the primal-dual greedy algorithm applied to the truncation (\ref{LP:T}) of a greedy system $(A,r)$.
 Then the cost of $x^*$ is no larger than $b \delta \text{OPT}$, if $r$ is non-negative, and $(b \delta + 1) \text{OPT}$, otherwise.
 Here, $b = 1$, if $\frac{r(S)^+ - r(\phi_e(S))^+}{a'_{S,e}} \in \ZZ_+$ for all $S \in \mathcal{L}$ and $e \in S$ with $a'_{S,e} > 0$, and $b = 2$, otherwise.
 \label{thm:sv:greedy-approximation}
\end{restatable}

\begin{proof}
In this version, we provide only a brief outline of the proof techniques for $b = 2$ with possibly negative rank.
The full proof can be found in the appendix.

Let $S_{\ell+1} \prec S_\ell \prec \dots \prec S_1$ be the dual chain constructed by the algorithm, where $r(S_{\ell+1}) \leq 0 < r(S_\ell)$ and let $e_1,\dots,e_\ell$ be the bottleneck elements.

Consider the lefthandside coefficients of any index $t$.
For element $e_j$ with index $t \leq j < \ell$, we can use $\delta$ to bound the coefficient in $A'$, as $r(\phi_{e_j}(S_t)) \geq r(S_{j+1}) > 0$.
This is true since $\mathcal{L}$ is modular, hence, $S_{j+1} = \phi_{e_j}(S_j) \preceq \phi_{e_j}(S_t)$.
\begin{align*}
a'_{S_t,e_j} x^*_{e_j}
\leq \delta a'_{S_j, e_j} x^*_{e_j}
= \delta a'_{S_j, e_j} \left\lceil \frac{r(S_j)^+ - r(S_{j+1})^+}{a'_{S_j,e_j}} \right\rceil
\leq 2 \delta \left( r(S_j) - r(S_{j+1}) \right).
\end{align*}
The first inequality is due to the definition of $\delta$.
The subsequent equality is due to the construction of $x^*$ in such a way that it covers the rank differences in each iteration.

Note that this argument does not necessarily hold for the final element $e_\ell$ as the definition of $\delta$ does not cover this element if $r(S_{\ell+1}) < 0$ and $a'_{S_\ell,e_\ell} < a_{S_\ell,e_\ell}$.
If $x^*_{e_\ell} = 1$, then $a'_{S_t,e_\ell} x^*_{e_\ell} = a'_{S_t,e_\ell} \leq r(S_t)$.

But $x^*_{e_\ell} > 1$ implies $a'_{S_{\ell},e_\ell} = a_{S_{\ell},e_\ell} < r(S_{\ell})$.
Hence, we can use $\delta$ and get
\begin{align*}
a'_{S_t,e_\ell} x^*_{e_\ell}
\leq \delta a'_{S_\ell,e_\ell} x^*_{e_\ell}
= \delta a'_{S_\ell,e_\ell} \left\lceil \frac{r(S_{\ell})^+ - r(S_{\ell+1})^+}{a'_{S_{\ell},e_\ell}} \right\rceil
\leq 2 \delta r(S_\ell).
\end{align*}
A simple union bound yields
$$a'_{S_t,e_\ell} x^*_{e_\ell} \leq 2 \delta r(S_\ell) + r(S_t).$$

Hence, for the constraint corresponding to $S_t$, we get:
\begin{align*}
\sum_{e \in S_t} a'_{S_t,e} x^*_e &= \sum_{j=t}^{\ell - 1} a'_{S_t,e_j} x^*_{e_j} + a'_{S_t,e_\ell} x^*_{e_\ell}
\leq \sum_{j=t}^{\ell - 1} 2 \delta \left( r(S_j) - r(S_{j+1}) \right) + 2 \delta r(S_\ell) + r(S_t) \\
&= 2 \delta \left( r(S_t) - r(S_{\ell}) \right) + 2 \delta r(S_\ell) + r(S_t)
= (2 \delta + 1)r(S_t).
\end{align*}
In the first equality in the second row, we use that the sum is telescopic.

This implies the following approximate complementary slackness conditions:
If $y^*_S > 0$, then $r(S) \leq a'_S x^* \leq (2 \delta + 1) r(S)$.
Moreover, the primal solution is constructed in such a way that $x^*_e > 0$ implies $\sum_{S \in \mathcal{L}} a'_{S,e} y^*_S = c_e$.
Hence, standard techniques for primal-dual approximation algorithms can be used to conclude the proof.
The other cases are proven analogously and can be found in the appendix.
\end{proof}

Hence, if $\delta$ is bounded by a small constant, we can show that the greedy algorithm obtains good solutions.
Note that the instance from Proposition \ref{prp:sv:linear-gap} shows that the integrality gap of a truncation can be of order $o(\log \delta)$.
We can also show that the analysis in Theorem \ref{thm:sv:greedy-approximation} is essentially tight.

\begin{corollary}
There exists a family of instances such that the truncation (\ref{LP:T}) of a greedy system $(A,r)$ has integrality gap $o(\log \delta)$.
\end{corollary}

\begin{restatable}{proposition}{prpSvGreedyLowerBound}
 The analysis in Theorem \ref{thm:sv:greedy-approximation} is tight up to constant factors.
 \label{prp:sv:greedy-bad-dual-solution}
\end{restatable}

\begin{proof}
 See appendix for the proof.
\end{proof}

To round up this section, we will see that all properties \ref{prop:first} - \ref{prop:last} are necessary in the sense that the removal of any of them results in a situation where the greedy algorithm does not provide a feasible solution.
While this does not rule out that other greedy algorithms may perform nicely, it points out certain limits of this analysis.

\begin{restatable}{proposition}{prpSvNecessityFeasible}
Suppose that a system $(A,r)$ satisfies \ref{prop:first} - \ref{prop:last} except for any one of the properties.
Then the greedy algorithm does not necessarily terminate with a feasible solution.
\label{prp:sv:greedy-necessity-feasible}
\end{restatable}

\begin{proof}
See appendix for the proof.
\end{proof}

\section{Generalization to multiple greedy systems}
\label{sec:product-version}

Although greedy systems already capture some well-known problems such as knapsack cover or optimization over contra-polymatroids, the modelling techniques are limited.
In this section, we discuss a generalization towards problems that are composed of multiple greedy systems on the same column set.
A full version of this section can be found in the appendix.
Due to space limitations, we discuss only a brief overview of the main results without proofs.

Again, let $E$ be the index set of columns.
Let $\mathcal{L}$ be a family of subsets of $E$ with some partial order $(\mathcal{L}, \preceq)$ associated.
Moreover, we will assume that $\mathcal{L}$ is actually a lattice with join $\vee$ and meet $\wedge$.
This family will be used similar to the previous section, it will also satisfy the properties elaborated in the previous section.
In particular, we will assume that the sets in $\mathcal{L}$ are pairwise different.
Moreover, let $\mathcal{U}$ be a family of subsets of $E$.
Multiple copies of the same subset are allowed in $\mathcal{U}$.
Let $\mathcal{B} = \mathcal{U} \times \mathcal{L}$,  $A \in \RR_+^{|\mathcal{B}| \times |E|}$ and $r: \mathcal{B} \rightarrow \RR$.
This time, the rows $a_{(U,S)}$ of matrix $A$ are indexed by tuples $(U,S) \in \mathcal{B}$.
The coefficients of matrix $A$ are denoted by $a_{(U,S),e}$ for row $(U,S) \in \mathcal{B}$ and column indexed by $e \in E$.
We will require $\{e \in E : a_{(U,S),e} > 0 \} = U \cap S$ for all $(U,S) \in \mathcal{B}$, that is, the support of each row of matrix $A$ indexed by a tuple $(U,S)$ equals the intersection of $U$ and $S$.
If $\mathcal{U} = \{E\}$, the situation from Section \ref{sec:simple-version} will be recovered.

We are interested in conditions of system $(A,r)$ such that problems of type
\begin{align*} \min_{x \in \ZZ_+^{|E|}} \left\{ c^T x \mid a_{(U,S)} x \geq r(U,S)\; \forall (U,S) \in \mathcal{B} \right\} \tag{P} \end{align*}
admit a bounded approximation guarantee via a simple primal-dual greedy algorithm.

In order to use the primal-dual greedy algorithm from Section \ref{sec:introduction}, we need an ordering on $\mathcal{B}$ which chooses the variables to be increased during the dual phase.
Note that we already assumed that $(\mathcal{L}, \preceq)$ is a partial order on $\mathcal{L}$.
We will assume that some additional partial orders are provided as follows.
For every $S \in \mathcal{L}$, let $(\mathcal{U}, \preceq_S)$ be a partial order of $\mathcal{U}$.
The partial orders are not required to be correlated in any way.
With these orderings, we compose the following lexicographic ordering for $\mathcal{B}$:
\begin{align*}
(U,S) \preceq_{\mathcal{B}} (U',S')\quad &\Leftrightarrow
\begin{aligned}
& \left( S \prec S' \right) \text{ or} \\
& \left( S = S' \text{ and } r(U,S) < r(U',S') \right) \text{ or} \\
& \left( S = S' \text{ and } r(U,S) = r(U',S') \text{ and } U \preceq_{S} U' \right)
\end{aligned}
\end{align*}

We use $(A,r)_{|U}$ to denote the subsystem of $(A,r)$ induced by fixing a set $U \in \mathcal{U}$.
Precisely, we say that $(A,r)_{|U} = (\bar{A}, \bar{r})$, where $\bar{A} \in \RR_+^{|\mathcal{L}| \times |E|}$ with coefficients $\bar{a}_{S,e} = a_{(U,S),e}$ and $\bar{r}: \mathcal{L} \rightarrow \RR, \bar{r}(S) = r(U,S)$.
Note that the ordering of $\mathcal{B}$ restricted to a subsystem is consistent in the way that a chain $(U_\ell,S_\ell) \preceq_{\mathcal{B}} \dots \preceq_{\mathcal{B}} (U_1,S_1)$ in $\mathcal{B}$ will induce a chain $S_\ell \preceq \dots \preceq S_1$ in every subsystem $(A,r)_{|U}$.
For any $e \in E$ we use the notation $\mathcal{B} \setminus \{e\} = \{ (U,S) \in \mathcal{B} : e \not\in S \}$ to denote the restriction of $\mathcal{B}$ to a subsystem of $\mathcal{L}$ that does not contain element $e$ in its support (with respect to the $\mathcal{L}$-component).
Note that the operation is assumed to have no effect on the $\mathcal{U}$ component, that is, we will observe tuples $(U,S) \in (\mathcal{B} \setminus \{e\})$ with $e \in U$.

The previous section elaborated that \ref{prop:first} - \ref{prop:last} are useful in order to prove approximation guarantees for subsystems $(A,r)_{|U}$.
In this section, we will assume that the restricted system $(A,r)_{|U}$ satisfies \ref{prop:first} - \ref{prop:last} for all $U \in \mathcal{U}$.
\begin{definition}
 A system $(A,r)$ on $\mathcal{B}$ (with respect to $(\mathcal{B}, \preceq_{\mathcal{B}})$) is called a \emph{greedy product system}, if for every $U \in \mathcal{U}$, the subsystem $(A,r)_{|U}$ satisfies \ref{prop:first} - \ref{prop:last}.
\end{definition}

Analogously to Section \ref{sec:simple-version}, we consider a truncated version of system $(A,r)$.
Otherwise, the integrality gap may be unbounded.
We apply the truncation from Definition \ref{def:sv:truncation} to each subsystem $(A,r)_{|U}, U \in \mathcal{U}$ individually and call the resulting system the \emph{truncation} of $(A,r)$.

\begin{restatable}{definition}{defPlTruncation}
Let $(A,r)$ be a greedy product system and define $A' \in \RR_+^{|\mathcal{B}| \times |E|}$ with coefficients as follows.
For $(U,S) \in \mathcal{B}$ and $e \in E$, set $a'_{(U,S),e} = \min\{a_{(U,S),e}, r(U,S)^+ - r(U,\phi_e(S))^+ \}$.
We call the system $(A',r)$ the \emph{truncation} of $(A,r)$.
\label{def:pl:truncation}
\end{restatable}

In this section, we will apply a revised version of the primal-dual greedy algorithm to system
\begin{align*}
 \min_{x \in \ZZ^{|E|}_+}  \{ c^T x \mid A'x \geq r \} \tag{T} \label{LP:T:product}
\end{align*}
and prove a bounded approximation guarantee similar to the previous section.

\sectionheadline{The revised primal-dual greedy algorithm}
In order to get results similar to Section \ref{sec:simple-version}, we need to slightly modify the greedy algorithm from Section \ref{sec:introduction}.
This time, we will combine the dual and primal phase in a single algorithm which is given in Figure \ref{alg:product:pseudocode}.

\begin{figure}[tb]
\begin{enumerate}
\item Initially, let $y^* \equiv 0, x^* \equiv 0$.
\item While $\mathcal{B} \neq \emptyset$
\begin{enumerate}
\item Let $B \subseteq \mathcal{B}$ be the maximal tuples in $\mathcal{B}$ with respect to ordering $\preceq_{\mathcal{B}}$. \label{alg:product:select-rows}
\item STOP if $r(U,S) \leq 0$ for $(U,S) \in B$. \label{alg:product:stop}
\item Raise $y^*_{(U,S)}$ for all $(U,S) \in B$ uniformly until some element $e^* \in E \setminus E^*$ becomes tight.
\item Let $S' = \phi_{e^*}(S)$ and set $x^*_{e^*} = \max \left\{ \left\lceil \frac{r(W,S)^+ - r(W,S')^+}{a'_{(W,S),e^*}} \right\rceil : W \in \mathcal{U}, a'_{(W,S),e^*} > 0 \right\}.$ \label{alg:product:set-x}
\item Add $e^*$ to $E^*$ and iterate with $\mathcal{B} = \mathcal{B} \setminus \{e^*\}$. \label{alg:product:lattice-subtract}
\end{enumerate}
\item For bottleneck elements $e^*$ in reverse order: Decrease $x_{e^*}$ as long as the solution remains feasible for all $(U,S) \in \mathcal{B}$. \label{alg:product:cleanup}
\end{enumerate}
\caption{Pseudocode of the revised primal-dual greedy algorithm.}
\label{alg:product:pseudocode}
\end{figure}

In contrast to Section \ref{sec:introduction}, we now increase the dual variable for \emph{all maximal} tuples $(U,S) \in \mathcal{B}$ with respect to $\preceq_{\mathcal{B}}$ uniformly.
Since $(\mathcal{L}, \preceq)$ is a lattice, all variables that are increased simultaneously during a single iteration share the same set $S \in \mathcal{L}$.
Moreover, by definition of the lexicographic order, they share the same rank value $r^*$.
If each partial order $(\mathcal{U}, \preceq_S)$ for $S \in \mathcal{L}$ exposes a single element, $\preceq_{\mathcal{B}}$ will also expose a single element.

We also adapt the construction of $x^*_{e^*}$ for bottleneck elements.
This time, we consider all rank differences $r(W,S)^+ - r(W,S')^+$ of sets $W \in \mathcal{U}$ and set $x^*_{e^*}$ sufficiently large as to cover \emph{all} these differences.
The element $S' = \phi_{e^*}(S) \in \mathcal{L}$ was chosen in such a way that it is the element $S'$ that is considered in the subsequent iteration of the main loop.
This will ensure primal feasibility.

Finally, we add an additional cleanup phase.
This will be beneficial, as variables from later iterations may render variables from previous iterations redundant.
In this case, we may carefully decrease variables in a post-processing step.
In general, deciding if a variable can be decreased by one may be a non-trivial task.
Moreover, determining the maximum in Line~\ref{alg:product:set-x}) is not simple, either.
In Table \ref{table:intro:result-summary} we provided some examples in which this is possible.

\sectionheadline{Approximation guarantee for the revised greedy algorithm}

Similar to Section \ref{sec:simple-version}, we can show that the truncation of a greedy product system does not cut off any integer feasible points.
Moreover, we can show that the greedy algorithm always obtains feasible primal solutions.
Due to space restrictions, we omit all feasibility results and provide only a summary regarding the approximability.

\begin{restatable}{lemma}{lemPlCleanupNecessary}
 Without the cleanup phase in Line~\ref{alg:product:cleanup}, an analysis similar to Theorem~\ref{thm:sv:greedy-approximation} for a greedy product system results in an approximation factor of at least $|E|$.
\end{restatable}

In order to characterize the influence on elements in terms of the cleanup phase, let us consider a solution $x^* \in \ZZ_+^{|E|}$ obtained by the revised greedy algorithm.
To get an intuition, let us assume for a second that the algorithm increased a single dual variable in each iteration, that is, Line~\ref{alg:product:select-rows}) returned a single maximum tuple in each iteration.
Let $(U_{\ell+1},S_{\ell+1}) \prec \dots \prec (U_1, S_1)$ be the constructed dual chain $e_i \in S_i \setminus S_{i+1}, 1 \leq i \leq \ell$ be the bottleneck elements.
As in the previous section, $r(U_{\ell+1},S_{\ell+1}) \leq 0 < r(U_{\ell},S_{\ell})$.

During the cleanup phase, the value $x^*_{e_i}$ of element $e_i$ was not further reduced because either $x^*_{e_i} = 0$, or there is at least one tuple $(U,S) \in \mathcal{B}$ such that
$$\sum_{f \in U \cap S} a'_{(U,S),f} x^*_f - a'_{(U,S),e_i} < r(U,S) \leq \sum_{f \in U \cap S} a'_{(U,S),f} x_f.$$
We call this tuple $(U,S)$ a \emph{witness} of bottleneck element $e_i$.
Note that this tuple was not necessarily considered in Line~\ref{alg:product:select-rows}).

But let us suppose that some element $e_i$ has a witness $(U_t, S_t)$ on the dual chain ($1 \leq t \leq \ell$).
Then $i \geq t$, otherwise $e_i \not\in S_t$.
The definition of witnesses implies
$$\sum_{f \in U_t \cap S_t} a'_{(U_t,S_t),f} x^*_f \leq 2 r(U_t,S_t).$$
In other words, if \emph{every} tuple $(U_t,S_t)$ of the dual chain was a witness for some element $e_i$, then $x^*$ would be a $2$-approximation for (\ref{LP:T}) by standard primal-dual approximation arguments (c.f.\ proof of Theorem~\ref{thm:sv:greedy-approximation}).
Of course, we can not expect this to happen in general.
But the following observation establishes a strong connection between witnesses and elements on the dual chain.
We will now cover the case that (possibly) multiple dual variables were increased simultaneously.

We define a \emph{multiplicity witness-cover} as follows.
Let $\mathcal{I} \subseteq \mathcal{B}$ be a family of tuples that were increased simultaneously in one iteration of the revised algorithm.
In this case, the rank value $r^*$ of all these tuples equals by definition of $(\mathcal{B}, \preceq_{\mathcal{B}})$.

We call $\mathcal{C} \subseteq \mathcal{B}$ a \emph{multiplicity witness-cover} of $\mathcal{I}$, if each tuple $(U,S) \in \mathcal{C}$ is a witness for some element $e \in E$, $r(U,S) \leq r^*$ and
every element $e \in E$ with $x^*_e > 0$ appears at least as often in $\mathcal{C}$, as it appears in $\mathcal{I}$.
That is, for all $e \in E$,
$$|\{ (U,S) \in \mathcal{I} : a'_{(U,S),e} x^*_e > 0 \}| \leq |\{ (U,S) \in \mathcal{C} : a'_{(U,S),e} x^*_e > 0 \}|.$$
If $\mathcal{C}$ is of small cardinality, we can show that $x^*$ is a good approximation.

We generalize our definition of $\delta$ from Section \ref{sec:simple-version} slightly to cover this case.
Given $(U,\emptyset), (W,S) \in \mathcal{B}, e \in E$, let
\begin{align*}
& \delta_{U,(W,S),e} = \begin{cases} \frac{a'_{(U,\emptyset),e}}{a'_{(W,S),e}} & a'_{(W,S),e} > 0 \text{ and } (r(W,\phi_e(S)) \geq 0 \text{ or } a'_{(W,S),e} = a_{(W,S),e}), \\ 1 & \text{otherwise}, \end{cases}
\end{align*}
and set $\delta = \max_{U,(W,S),e} \left\{ \delta_{U,(W,S),e} \right\}$.
The following Theorem \ref{thm:ps:greedy-approximation} yields bounds on the solution cost, depending on $A$ being binary or a general matrix.
Finally, Proposition \ref{prp:ps:greedy-bad-dual-solution} shows that the dependency on $k$ is inherent in the type of dual solution constructed.

\begin{restatable}{theorem}{thmPsMultiDualGreedyApproximation}
 Let $(A,r)$ be a greedy product system and let $k \in \ZZ_+$.
 Let $(x^*,y^*)$ be the solution obtained by the revised greedy algorithm with dual support $\mathcal{I}_i \subseteq \mathcal{B}$ in iteration $i$.
 If each family $\mathcal{I}_i$ has a witness cover of size at most $k |\mathcal{I}_i|$, then $x^*$ has cost no larger than $k (\delta + 1) OPT$.

 If, additionally, the truncation $A'$ is a binary matrix, then the solution has cost bounded by $k OPT$.
 \label{thm:ps:greedy-approximation}
\end{restatable}

\begin{restatable}{proposition}{prpPsGreedyLowerBound}
 For every $k \in \ZZ_+$ there is a greedy product system $(A,r)$ with truncation $A' \in \{0,1\}^{|\mathcal{B}| \times |E|}$ such that the dual $y^*$ obtained by the revised greedy algorithm has optimality gap $k$.
 \label{prp:ps:greedy-bad-dual-solution}
\end{restatable}

\section*{Appendix}

\subsection*{Proofs from Section \ref{sec:simple-version}}

\svObsPhiOrderPreserving*

\begin{proof}
 We will show that $\phi_e(S) = S \wedge \phi_e(T)$.
 \begin{align*}
 S \wedge \phi_e(T) &= \sup\{ L \in \mathcal{L} : L \preceq S, \phi_e(T) \} = \sup\{ L \in \mathcal{L} : L \preceq S, T, e \not\in L \} \\
 &= \sup\{ L \in \mathcal{L} : L \preceq S, e \not\in L \} = \phi_e(S).
 \end{align*}
\end{proof}

\svObsSublattice*

\begin{proof}
 In order to prove the lemma, we need to ensure that \ref{prop:lattice-modular} is still satisfied.
 The remaining properties will clearly hold.

 First of all, let us ensure that $\mathcal{L}'$ will be a lattice.
 Suppose this was not the case and let $S,T \in \mathcal{L}'$ be two maximum elements.
 Then $S \vee_{\mathcal{L}} T \not\in \mathcal{L}'$.
 This implies that $e \in S \vee_{\mathcal{L}} T$, which contradicts to the last statement in \ref{prop:lattice-modular}.

 It remains to show that $\mathcal{L}'$ is modular.
 Suppose it was not modular, then by \cite{birkhoff1948lattice} Theorem 12, we know that $\mathcal{L}'$ contains a sublattice isomorphic to $N_5$.
 Note that $N_5$ is a $5$-cycle.
 For ease of notation, let $0,1$ denote the minimum and maximum elements in $N_5$ and let $U,S,T$ be the non-extreme elements such that $0 \prec S \prec T \prec 1$ and $U$ and $S,T$ are incomparable.
 Since $\mathcal{L}$ is modular, there are two cases (with respect to $\mathcal{L}$).
 \begin{itemize}
  \item $0 \prec S \wedge U \preceq T \wedge U \prec T,U$ with $S \wedge U, T \wedge U \not\in \mathcal{L}'$.
  This implies that $e \in S \wedge U$ but $e \not\in T$, contradicting to \ref{prop:A-monotone}.
  \item $S,U \prec S \vee U \preceq T \vee U \prec 1$ with $S \vee U, T \vee U \not\in \mathcal{L}'$.
  Again, $e \in T \vee U$ and $e \not\in 1$.
  Hence, a contradiction to \ref{prop:A-monotone}.
 \end{itemize}
 Hence, $\mathcal{L}'$ is modular.
 The proof is concluded.
\end{proof}

\obsSvMarginalIncreaseSupermodularity*

\begin{proof}
 Let us apply \ref{prop:A-r-coupling} with $\bar{S} = \phi_e(T)$ and $\bar{T} = S$.
 Due to modularity, we have $\bar{S} \wedge \bar{T} = \phi_e(S)$ and $\bar{S} \vee \bar{T} = T$.
 Moreover, $e \in \bar{T} \setminus (\bar{S} \wedge \bar{T}) = S \setminus \phi_e(S)$.
 Hence,

 \begin{align*}
  & \frac{r(\bar{T}) - r(\bar{S} \wedge \bar{T})}{a_{\bar{T},e}} \leq \frac{r(\bar{S} \vee \bar{T}) - r(\bar{S})}{a_{\bar{S} \vee \bar{T},e}}
  \Leftrightarrow \frac{r(S) - r(\phi_e(S))}{a_{S,e}} \leq \frac{r(T) - r(\phi_e(T))}{a_{T,e}}. \qedhere
 \end{align*}
\end{proof}

\obsSvMonotoneAroundZero*

\begin{proof}
 We will derive this property from the statement of Lemma \ref{lem:sv:marginal-increase-supermodularity}.
 We will show the following:
 \begin{align*}
   \frac{ r(S) - r(\phi_e(S)) }{ a_{S,e} } \leq \frac{ r(T) - r(\phi_e(T)) }{ a_{T,e} }
   \Leftrightarrow \frac{ r(S) }{ a_{S,e} } - \frac{ r(T) }{ a_{T,e} } \leq \frac{ r(\phi_e(S)) }{ a_{S,e} } - \frac{ r(\phi_e(T)) }{ a_{T,e} } \overset{!}{\leq} 0.
 \end{align*}
 The proof is almost concluded by reformulating the terms:
 \begin{align*}
   \frac{ r(\phi_e(S)) }{ a_{S,e} } - \frac{ r(\phi_e(T)) }{ a_{T,e} } \leq 0
   \Leftrightarrow \frac{ a_{T,e} }{ a_{S,e} } r(\phi_e(S)) \leq r(\phi_e(T)).
 \end{align*}
 The proof is concluded.
 Note that $\phi_e(S) \preceq \phi_e(T)$ by modularity of the lattice.
 Hence, \ref{prop:r-monotone} implies $r(\phi_e(S)) \leq r(\phi_e(T))$.
 Moreover, by \ref{prop:A-monotone}, $\frac{ a_{T,e} }{ a_{S,e} } \geq 1$.
 Since both, $r(\phi_e(S)), r(\phi_e(T)) \leq 0$, the result follows.
\end{proof}

\lemSvTruncationFeasible*

\begin{proof}
 First, let us show that $A'$ is monotone increasing.
 Therefore, let us consider $S \preceq T \in \mathcal{L}$ with $e \in S$.
 We will show $r(S)^+ - r(\phi_e(S))^+ \leq r(T)^+ - r(\phi_e(T))^+$.
 If $r(\phi_e(S)) \geq 0$, this is certainly true due to Lemma \ref{lem:sv:marginal-increase-supermodularity} and \ref{prop:A-monotone}.
 If both, $r(\phi_e(S)) < 0$ and $r(\phi_e(T)) < 0$, the result is implied by monotonicity of $r$.
 Finally, if $r(\phi_e(S)) < 0 \leq r(\phi_e(T))$, we have
 \begin{align*}
 r(T)^+ - r(\phi_e(T))^+
 &\geq \frac{a_{S,e}}{a_{T,e}} \left( r(T)^+ - r(\phi_e(T))^+ \right)
 = \frac{a_{S,e}}{a_{T,e}} \left( r(T) - r(\phi_e(T)) \right) \\
 & \overset{Lem. \ref{lem:sv:marginal-increase-supermodularity}}{\geq} \frac{a_{S,e}}{a_{S,e}} \left( r(S) - r(\phi_e(S)) \right) > r(S)^+ - r(\phi_e(S))^+.
 \end{align*}
 Since $A'$ is the minimum of two monotone increasing functions, it is also monotone increasing.

 Suppose that $x \in (\ref{LP:T})$.
 By definition of the truncation, we get
 $$r(S) \leq \sum_{e \in S} a'_{S,e} x_e \leq \sum_{e \in S} a_{S,e} x_e$$
 for all $S \in \mathcal{L}$.
 Hence, $x \in (\ref{LP:P})$.

 Now, let us assume that $x \in (\ref{LP:P})$ and let $S$ be a row violated in $(\ref{LP:T})$ that has minimum cardinality among all such rows.
 Let $e \in S$ with $x_e > 0$ and $a_{S,e} > a'_{S,e} = r(S)^+ - r(S')^+$, where $S' = \phi_e(S)$.
 If $r(S')^+ = 0$, then $x_e > 0$ implies $a'_{S,e} x_e \geq r(S)$, hence, we assume that $r(S')^+ > 0$.
 If there was no such element, all coefficients of positive variables in row $S$ in (\ref{LP:P}) and (\ref{LP:T}) coincide, hence, the constraint can not be violated.
 By minimality of $S$, we have $\sum_{f \in S'} a'_{S',f} x_f \geq r(S').$
 Moreover,
 \begin{align*}
 \sum_{f \in S'} a'_{S',f} x_f
 \leq \sum_{f \in S'} a'_{S,f} x_f
 = \sum_{f \in S} a'_{S,f} x_f - \sum_{f \in S \setminus S'} a'_{S,f} x_f
 \leq \sum_{f \in S} a'_{S,f} x_f - a'_{S,e} x_e.
 \end{align*}
 Hence, we can conclude
 \begin{align*}
  & r(S') \leq \sum_{f \in S'} a'_{S',f} x_f
  \leq \sum_{f \in S} a'_{S,f} x_f - a'_{S,e} x_e
  \leq \sum_{f \in S} a'_{S,f} x_f - (r(S) - r(S')) \\
  & \Leftrightarrow r(S) \leq \sum_{f \in S} a'_{S,f} x_f,
 \end{align*}
 which is a contradiction.
\end{proof}

\lemRfGreedyTruncationFeasible*

\begin{proof}
Let us assume that the greedy algorithm terminated with a vector $x^* \in \ZZ_+^{|E|}$ with chain of dual variables $S_{\ell + 1} \prec \dots \prec S_1 = E$ and positive variables $x^*_{e_i}$ for bottleneck elements $e_1,\dots,e_\ell$.
For $1 \leq i \leq \ell$, the algorithm sets $x^*_{e_i}$ such that it satisfies
$$x^*_{e_i} \geq \left\lceil \frac{r(S_i)^+ - r(S_{i+1})^+}{a'_{S_i,e_i}} \right\rceil \geq \left\lceil \frac{r(S_i)^+ - r(S_{i+1})^+}{a_{S_i,e_i}} \right\rceil.$$
Note that the proof of Lemma~\ref{lem:sv:set-feasible} did not rely on the fact that the dual solution $y^*$ was constructed by the greedy algorithm.
In fact, the proof shows feasibility, whenever the chain can be constructed via the function $\phi$, the final element of the chain $S_{\ell+1}$ has non-positive rank and the primal vector $x^*$ was chosen as in the primal phase of the algorithm.
Hence, Lemma~\ref{lem:sv:set-feasible} is applicable and $x^*$ is feasible for (\ref{LP:P}).
By Lemma \ref{lem:sv:truncation-feasible}, the solution is also feasible for (\ref{LP:T}), which concludes the proof.
\end{proof}

\prpSvLogInapproximable*

\begin{proof}
Let us consider a subset cover instance with groundset $G$ and subsets $U_i \subseteq G$ with cost $c_i$, $1 \leq i \leq n$. Let $|G| = M$.
The goal is to find a collection of subsets of minimum cost which covers all elements.
We will show that the following formulation is valid for subset cover:
\begin{align}
\min_{x \in \ZZ_+^n} \left\{ c^T x \mid \sum_{i \in S} |U_i| x_i \geq r(S) \; \forall S \subseteq [n] \right\} \tag{S} \label{LP:set-cover}
\end{align}
with $r(S) = M - | \cup_{i \not\in S} U_i |$.
Moreover, we will see that it satisfies \ref{prop:first} - \ref{prop:last}.

First, let us evaluate the properties with respect to the Boolean lattice.
Since the coefficients of $A$ are constant, \ref{prop:A-monotone} is satisfied.
Moreover, \ref{prop:r-monotone} clearly also holds and, since the coefficients in every column take on a single non-zero value.
Finally, let us evaluate that $r$ is supermodular.
Hence, \ref{prop:A-r-coupling} will hold.
Therefore, let $S, T \subseteq G$ and let $e \in T \setminus (S \cap T)$.
Then $a_{T,e} = a_{S \cup T,e}$ and
\begin{align*}
 r(T) - r(S \cap T) \leq r(S \cup T) - r(S)
 &\Leftrightarrow | \cup_{i \not\in S \cap T} U_i | - | \cup_{i \not\in T} U_i | \leq | \cup_{i \not\in S} U_i |- | \cup_{i \not\in S \cup T} U_i | \\
 &\Leftrightarrow | \cup_{i \not\in S \cap T} U_i | + | \cup_{i \not\in S \cup T} U_i | \leq | \cup_{i \not\in S} U_i | + | \cup_{i \not\in T} U_i |
\end{align*}

Now, let us verify that (\ref{LP:set-cover}) is a valid formulation for subset cover.
Let $x$ be the incidence vector of a feasible subset cover solution and let $S$ be any constraint with positive rank, then
  $$ \sum_{i \in S} |U_i| x_i + | \cup_{i \not\in S} U_i | \geq | \cup_{i : x_i = 1} U_i| \geq M $$
  and for the constraint, we get
  \begin{align*}
    \sum_{i \in S} |U_i| x_i \geq M - | \cup_{i \not\in S} U_i |
    \Leftrightarrow \sum_{i \in S} |U_i| x_i + | \cup_{i \not\in S} U_i | \geq M.
  \end{align*}
Finally, every feasible solution $x$ in (\ref{LP:set-cover}) is also a feasible subset cover.
Suppose not, then $| \cup_{i: x_i = 1} U_i| < M$ and $r(E \setminus \{i : x_i = 1\}) > 0$.
But $x$ has zero lefthandside value for this constraint which contradicts to the feasibility of $x$ in (\ref{LP:set-cover}).

Since both objective functions coincide, an $\alpha$-approximation for the truncation (\ref{LP:T}) yields an $\alpha$-approximation for (\ref{LP:set-cover}).
The proof is concluded as there is no $(1-o(1))\log n$ approximation for subset cover unless $NP \subseteq DTIME(n^{O(\log\log n)})$ \cite{feige1998threshold}.
\end{proof}

\prpSvLinearGap*

\begin{proof}
For a given number $n \in \ZZ_+$, let us define $n$ elements $E = \{1,\dots,n\}$.
We will consider the Boolean lattice $\mathcal{L} = 2^E$ and an according system $(A,r)$ with $A \in \ZZ_+^{|\mathcal{L}| \times |E|}$
with $a_{S,e} = 2^n$, if $e \in S$ and zero otherwise with cost $c \equiv 1$.
Define $r(S) = 2^n(2^{-(n-|S|)} - 2^{-\frac{n}{2}})$ for all $S \subseteq E$.

Let $x^*$ be any feasible solution to (\ref{LP:T}).
Since for any subset of elements of cardinality at least $\frac{n}{2} + 1$, the rank is positive, an optimum solution has cost at least $\frac{n}{2}$.
Now, let us consider the fractional solution $x \equiv \frac{4}{n-1}$.
For any row $S$ with cardinality $k > \frac{n}{2}$, we have:
\begin{align*} \sum_{e \in S} a'_{S,e} x^*_e &= k \frac{4}{n-1} 2^n(2^{-(n-k)} - 2^{-\frac{n}{2}} - 2^{-(n-(k+1))} + 2^{-\frac{n}{2}}) \\
&\geq 2 \cdot 2^n 2^{-(n-(k+1))} \geq 2^n 2^{- (n-k)} \geq r(S)
\end{align*}
If $S$ is of smaller cardinality, $r(S) \leq 0$.
Hence, $x^*$ is feasible.
Moreover, the cost $c^T x^* = \frac{4n}{n-1} \leq 4$.
\end{proof}

\thmSvGreedyApproximation*

\begin{proof}
Let $S_{\ell+1} \prec S_\ell \prec \dots \prec S_1$ be the dual chain constructed by the algorithm, where $r(S_{\ell+1}) \leq 0 < r(S_\ell)$ and let $e_1,\dots,e_\ell$ be the bottleneck elements.

Consider the lefthandside coefficients of any index $t$.
For $b = 2$ and for element $e_j$ with index $t \leq j < \ell$, we can use $\delta$ to bound the coefficient in $A'$, as $r(\phi_{e_j}(S_t)) \geq r(S_{j+1}) > 0$.
This is true since $\mathcal{L}$ is modular, hence, $S_{j+1} = \phi_{e_j}(S_j) \preceq \phi_{e_j}(S_t)$.
\begin{align*}
a'_{S_t,e_j} x^*_{e_j}
\leq \delta a'_{S_j, e_j} x^*_{e_j}
= \delta a'_{S_j, e_j} \left\lceil \frac{r(S_j)^+ - r(S_{j+1})^+}{a'_{S_j, e_j}} \right\rceil
\leq 2 \delta \left( r(S_j) - r(S_{j+1}) \right).
\end{align*}
The first inequality is due to the definition of $\delta$.
The subsequent equality is due to the construction of $x^*$ in such a way that it covers the rank differences in each iteration.

If $b = 1$, we analogously get:
\begin{align*}
a'_{S_t,e_j} x^*_{e_j}
\leq \delta a'_{S_j, e_j} x^*_{e_j}
= \delta a'_{S_j, e_j} \left\lceil \frac{r(S_j)^+ - r(S_{j+1})^+}{a'_{S_j, e_j}} \right\rceil
= \delta \left( r(S_j) - r(S_{j+1}) \right).
\end{align*}

Note that this argument does not necessarily hold for the final element $e_\ell$ as the definition of $\delta$ does not cover this element if $r(S_{\ell+1}) < 0$ and $a'_{S_\ell,e_\ell} < a_{S_\ell,e_\ell}$.
If $x^*_{e_\ell} = 1$, then $a'_{S_t,e_\ell} x^*_{e_\ell} = a'_{S_t,e_\ell} \leq r(S_t)$.

But $x^*_{e_\ell} > 1$ implies $a'_{S_{\ell},e_\ell} = a_{S_{\ell},e_\ell} < r(S_{\ell})$.
Hence, we can use $\delta$ and get
\begin{align*}
a'_{S_t,e_\ell} x^*_{e_\ell}
\leq \delta a'_{S_\ell,e_\ell} x^*_{e_\ell}
= \delta a'_{S_\ell,e_\ell} \left\lceil \frac{r(S_{\ell})^+ - r(S_{\ell+1})^+}{a'_{S_{\ell},e_\ell}} \right\rceil
\leq 2 \delta r(S_\ell).
\end{align*}

And analogously, if $b = 1$, we have:
\begin{align*}
a'_{S_t,e_\ell} x^*_{e_\ell}
\leq \delta a'_{S_\ell,e_\ell} x^*_{e_\ell}
= \delta a'_{S_\ell,e_\ell} \left\lceil \frac{r(S_{\ell})^+ - r(S_{\ell+1})^+}{a'_{S_{\ell},e_\ell}} \right\rceil
= \delta r(S_\ell).
\end{align*}

A simple union bound yields
$$a'_{S_t,e_\ell} x^*_{e_\ell} \leq b \delta r(S_\ell) + r(S_t).$$

Hence, for the constraint corresponding to $S_t$, we get:
\begin{align*}
\sum_{e \in S_t} a'_{S_t,e} x^*_e &= \sum_{j=t}^{\ell - 1} a'_{S_t,e_j} x^*_{e_j} + a'_{S_t,e_\ell} x^*_{e_\ell}
\leq \sum_{j=t}^{\ell - 1} b \delta \left( r(S_j) - r(S_{j+1}) \right) + b \delta r(S_\ell) + r(S_t) \\
&= b \delta \left( r(S_t) - r(S_{\ell}) \right) + b \delta r(S_\ell) + r(S_t)
\leq (b \delta + 1)r(S_t).
\end{align*}
In the first equality in the second row, we use that the sum is telescopic.

Now, if $r$ is non-negative, we have $r(S_{\ell+1}) = 0$, hence, the inequalities we derived for $t \leq j < \ell$ also hold for index $\ell$, that is, for the final element.
Recall that this is due to the fact that, in this case, the definition of $\delta$ also applies for the element $e_\ell$.
Hence, similarly, we have
\begin{align*}
\sum_{e \in S_t} a'_{S_t,e} x^*_e &= \sum_{j=t}^{\ell} a'_{S_t,e_j} x^*_{e_j}
\leq \sum_{j=t}^{\ell} b \delta \left( r(S_j) - r(S_{j+1}) \right)
= b \delta \left( r(S_t) - r(S_{\ell+1}) \right)
= b \delta r(S_t).
\end{align*}

For the remaining part, let $a = 0$, if $r$ is non-negative, and $a = 1$, otherwise.
This implies the following approximate complementary slackness conditions:
If $y^*_S > 0$, then $r(S) \leq a'_S x^* \leq (b \delta + a) r(S)$.
Here, $a'_S$ denotes the row of $A'$ induced by index $S$.
Moreover, the vector $x^*$ has positive coefficients only for tight elements, that is, $x^*_e > 0$ implies $\sum_{S \in \mathcal{L}} a'_{S,e} y^*_S = c_e$.
Hence, the total cost of the solution can be rewritten as:
\begin{align*}
\sum_{e \in E} x^*_e c_e
= \sum_{e \in E} x^*_e \sum_{S \in \mathcal{L}} y^*_S a'_{S,e}
= \sum_{S \in \mathcal{L}} y^*_S \sum_{e \in E} a'_{S,e} x^*_e
\leq (b \delta + a) \sum_{S \in \mathcal{L}} y^*_S r(S)
\leq (b \delta + a) \text{OPT}.
\end{align*}
The latter step is due to $y^*$ being feasible for the dual of the relaxation of $(\ref{LP:T})$.
In other words, it is a lower bound on the optimal objective function value.
\end{proof}

\prpSvGreedyLowerBound*

\begin{proof}
We will show the following result:
For every $\delta \in \ZZ_+$ there is a greedy system $(A,r)$ such that
 \begin{enumerate}
  \item the truncation (\ref{LP:T}) does not have an integrality gap and
  \item the dual solution obtained by the greedy algorithm has an optimality gap of $\delta$.
 \end{enumerate}

 Let us consider the following type of subset cover instance.
 $G = \{1,\dots,M+n\}$ with subsets $U_i = \{1,\dots,M,M+i\}$ for $1 \leq i \leq n$.
 That is, subset $U_i$ contains the first $M$ elements and, additionally, element $M+i$.
 Moreover, let $c_i = M+1$ for all subsets.

 Let us consider the formulation from Proposition \ref{prp:sv:log-inapproximable}, that is,
 $a_{S,i} = |U_i| = M+1$ and $r(S) = M + n - |\cup_{i \not\in S} U_i|$.
 Note that $r(E) = M + n$ and $r(S) = |S|$, otherwise.
 Then the truncation has coefficients
 $$ a'_{S,i} = \begin{cases} M+1, & S = E \\ 1, & |S| < n \text{ and } i \in S \\ 0, & S : i \in S.  \end{cases}$$
 In this instance, we have $\delta = M+1$.
 The greedy algorithm will increase $y_E$ until some element becomes tight.
 In this case, $y_E = 1$ and all elements become tight simultaneously.
 Hence, this is the solution the algorithm computes with dual objective value $M+n$, which will be the lower bound for the primal solution.
 Since the only feasible primal solution is $x \equiv 1$ with cost $n(M+1)$, the gap between the two is of order $n$.

 Note that the optimum dual solution matches this bound.
 This can be reached, if we set $y_S = \frac{M+1}{n-1}$ for all $S = E \setminus \{i\}$, $1 \leq i \leq n$.
 For every element $i$, we will have $\sum_S y_S a'_{S,i} = \sum_{j \neq i} y_{E \setminus \{j\}} = M + 1$.
 Hence, it is feasible.
 The dual objective value is $\frac{M+1}{n-1} n (n - 1) = (M + 1) n$.

 Hence, the instance does not have an integrality gap, but the gap between a dual solution on a chain and an optimum dual solution can be of order $n$.
 if we set $n = M+1$, the gap will be $\delta$.
\end{proof}

\prpSvNecessityFeasible*

\begin{proof}
For each case, we will provide systems $(A,r)$.
\begin{itemize}
 \item Let us suppose that $(A,r)$ satisfies \ref{prop:first} - \ref{prop:last} except for \ref{prop:r-monotone}.
The following system shows that the greedy algorithm does not terminate with a feasible solution.
\begin{align*}
 \min \quad & 2 x_1 + x_2 \\
 \text{s.t.} \quad & x_1 + x_2 \geq 2 \\
 & x_1 \geq -2 \\
 & x_2 \geq 1 \\
 & 0 \geq 0
\end{align*}
The greedy algorithm will set $y^*_E = 1$ at which point element 2 becomes tight.
It will set $x^*_2 = 2$ and terminate, as $r(\{2\}) = -2 \leq 0$.
For $S = \{2\}$, the solution will be infeasible.

\item Next, let us assume that \ref{prop:A-monotone} is not satisfied.
\begin{align*}
 \min \quad & x_1 + 6 x_2 \\
 \text{s.t.} \quad & x_1 + x_2 \geq 10 \\
 & 5 x_1 \geq 5 \\
 & 5 x_2 \geq 5 \\
 & 0 \geq 0
\end{align*}
The algorithm will set $y^*_E = 1$.
At this point, it will set $x^*_1 = 5$.
Next, $y^*_{\{1\}} = 1$ and $x^*_2 = 1$.
The solution vector $(5,1)$ is not feasible for $S = E$.

\item Following, let us assume that $(A,r)$ does not satisfy \ref{prop:A-r-coupling}.
\begin{align*}
 \min \quad & 5 x_1 + 6 x_2 \\
 \text{s.t.} \quad & 5 x_1 + 5 x_2 \geq 10 \\
 & 2.5 x_1 \geq 5 \\
 & 2.5 x_2 \geq 5 \\
 & 0 \geq 0
\end{align*}
The algorithm will start with $y^*_E = 1$ and set $x^*_1 = 1$ in order to cover the rank difference $r(E) - r(\{2\})$.
Afterwards, it will set $y^*_{\{1\}} = 1$.
At this point, we set $x^*_2 = 5$.
Hence, the algorithm terminates with the solution vector $(1,2)$, which is infeasible for $S = \{1\}$.
\end{itemize}
\end{proof}

\subsection*{Generalization to multiple greedy systems (full version)}
\label{sec:product-version-full}

Although greedy systems already capture some well-known problems such as knapsack cover or optimization over contra-polymatroids, the modelling techniques are limited.
In this section, we discuss a generalization towards problems that are composed of multiple greedy systems on the same column set.

Again, let $E$ be the index set of columns.
Let $\mathcal{L}$ be a family of subsets of $E$ with some partial order $(\mathcal{L}, \preceq)$ associated.
Moreover, we will assume that $\mathcal{L}$ is actually a lattice with join $\vee$ and meet $\wedge$.
This family will be used similar to the previous section, it will also satisfy the properties elaborated in the previous section.
In particular, we will assume that the sets in $\mathcal{L}$ are pairwise different.
Moreover, let $\mathcal{U}$ be a family of subsets of $E$.
Multiple copies of the same subset are allowed in $\mathcal{U}$.
Let $\mathcal{B} = \mathcal{U} \times \mathcal{L}$,  $A \in \RR_+^{|\mathcal{B}| \times |E|}$ and $r: \mathcal{B} \rightarrow \RR$.
This time, the rows $a_{(U,S)}$ of matrix $A$ are indexed by tuples $(U,S) \in \mathcal{B}$.
The coefficients of matrix $A$ are denoted by $a_{(U,S),e}$ for row $(U,S) \in \mathcal{B}$ and column indexed by $e \in E$.
We will require $\{e \in E : a_{(U,S),e} > 0 \} = U \cap S$ for all $(U,S) \in \mathcal{B}$, that is, the support of each row of matrix $A$ indexed by a tuple $(U,S)$ equals the intersection of $U$ and $S$.
If $\mathcal{U} = \{E\}$, the situation from Section \ref{sec:simple-version} will be recovered.

We are interested in conditions of system $(A,r)$ such that problems of type
\begin{align*} \min_{x \in \ZZ_+^{|E|}} \left\{ c^T x : a_{(U,S)} x \geq r(U,S)\; \forall (U,S) \in \mathcal{B} \right\} \tag{P} \label{LP:P:product-full} \end{align*}
admit a bounded approximation guarantee via a simple primal-dual greedy algorithm.

In order to use the primal-dual greedy algorithm from Section \ref{sec:introduction}, we need an ordering on $\mathcal{B}$ which chooses the variables to be increased during the dual phase.
Note that we already assumed that $(\mathcal{L}, \preceq)$ is a partial order on $\mathcal{L}$.
We will assume that some additional partial orders are provided as follows.
For every $S \in \mathcal{L}$, let $(\mathcal{U}, \preceq_S)$ be a partial order of $\mathcal{U}$.
The partial orders are not required to be correlated in any way.
With these orderings, we compose the following lexicographic ordering for $\mathcal{B}$:
\begin{align*}
(U,S) \preceq_{\mathcal{B}} (U',S')\quad &\Leftrightarrow
\begin{aligned}
& \left( S \prec S' \right) \text{ or} \\
& \left( S = S' \text{ and } r(U,S) < r(U',S') \right) \text{ or} \\
& \left( S = S' \text{ and } r(U,S) = r(U',S') \text{ and } U \preceq_{S} U' \right)
\end{aligned}
\end{align*}

We use $(A,r)_{|U}$ to denote the subsystem of $(A,r)$ induced by fixing a set $U \in \mathcal{U}$.
Precisely, we say that $(A,r)_{|U} = (\bar{A}, \bar{r})$, where $\bar{A} \in \RR_+^{|\mathcal{L}| \times |E|}$ with coefficients $\bar{a}_{S,e} = a_{(U,S),e}$ and $\bar{r}: \mathcal{L} \rightarrow \RR, \bar{r}(S) = r(U,S)$.
Note that the ordering of $\mathcal{B}$ restricted to a subsystem is consistent in the way that a chain $(U_\ell,S_\ell) \preceq_{\mathcal{B}} \dots \preceq_{\mathcal{B}} (U_1,S_1)$ in $\mathcal{B}$ will induce a chain $S_\ell \preceq \dots \preceq S_1$ in every subsystem $(A,r)_{|U}$.
For any $e \in E$ we use the notation $\mathcal{B} \setminus \{e\} = \{ (U,S) \in \mathcal{B} : e \not\in S \}$ to denote the restriction of $\mathcal{B}$ to a subsystem of $\mathcal{L}$ that does not contain element $e$ in its support (with respect to the $\mathcal{L}$-component).
Note that the operation is assumed to have no effect on the $\mathcal{U}$ component, that is, we will observe tuples $(U,S) \in (\mathcal{B} \setminus \{e\})$ with $e \in U$.

The previous section elaborated that \ref{prop:first} - \ref{prop:last} are useful in order to prove approximation guarantees for subsystems $(A,r)_{|U}$.
In this section, we will assume that the restricted system $(A,r)_{|U}$ satisfies \ref{prop:first} - \ref{prop:last} for all $U \in \mathcal{U}$.
\begin{definition}
 A system $(A,r)$ on $\mathcal{B}$ (with respect to $(\mathcal{B}, \preceq_{\mathcal{B}})$) is called a \emph{greedy product system}, if for every $U \in \mathcal{U}$, the subsystem $(A,r)_{|U}$ satisfies \ref{prop:first} - \ref{prop:last}.
\end{definition}

Analogously to Section \ref{sec:simple-version}, we consider a truncated version of system $(A,r)$.
Otherwise, the integrality gap may be unbounded.
We apply the truncation from Definition \ref{def:sv:truncation} to each subsystem $(A,r)_{|U}, U \in \mathcal{U}$ individually and call the resulting system the \emph{truncation} of $(A,r)$.

\defPlTruncation*

In this section, we will apply a revised version of the primal-dual greedy algorithm to system
\begin{align*}
 \min_{x \in \ZZ^{|E|}_+}  \{ c^T x \mid A'x \geq r \} \tag{T} \label{LP:T:product-full}
\end{align*}
and prove a bounded approximation guarantee similar to the previous section.

\sectionheadline{The revised primal-dual greedy algorithm}
In order to get results similar to Section \ref{sec:simple-version}, we need to slightly modify the greedy algorithm from Section \ref{sec:introduction}.
This time, we will combine the dual and primal phase in a single algorithm which is given in Figure \ref{alg:product:pseudocode}.

\begin{figure}[tb]
\begin{enumerate}
\item Initially, let $y^* \equiv 0, x^* \equiv 0$.
\item While $\mathcal{B} \neq \emptyset$
\begin{enumerate}
\item Let $B \subseteq \mathcal{B}$ be the maximal tuples in $\mathcal{B}$ with respect to ordering $\preceq_{\mathcal{B}}$. \label{appendix:alg:product:select-rows}
\item STOP if $r(U,S) \leq 0$ for $(U,S) \in B$. \label{appendix:alg:product:stop}
\item Raise $y^*_{(U,S)}$ for all $(U,S) \in B$ uniformly until some element $e^* \in E \setminus E^*$ becomes tight.
\item Let $S' = \phi_{e^*}(S)$ and set $x^*_{e^*} = \max \left\{ \left\lceil \frac{r(W,S)^+ - r(W,S')^+}{a'_{(W,S),e^*}} \right\rceil : W \in \mathcal{U}, a'_{(W,S),e^*} > 0 \right\}.$ \label{appendix:alg:product:set-x}
\item Add $e^*$ to $E^*$ and iterate with $\mathcal{B} = \mathcal{B} \setminus \{e^*\}$. \label{appendix:alg:product:lattice-subtract}
\end{enumerate}
\item For bottleneck elements $e^*$ in reverse order: Decrease $x_{e^*}$ as long as the solution remains feasible for all $(U,S) \in \mathcal{B}$. \label{appendix:alg:product:cleanup}
\end{enumerate}
\caption{Pseudocode of the revised primal-dual greedy algorithm.}
\label{appendix:alg:product:pseudocode}
\end{figure}

In contrast to Section \ref{sec:introduction}, we now increase the dual variable for \emph{all maximal} tuples $(U,S) \in \mathcal{B}$ with respect to $\preceq_{\mathcal{B}}$ uniformly.
Since $(\mathcal{L}, \preceq)$ is a lattice, all variables that are increased simultaneously during a single iteration share the same set $S \in \mathcal{L}$.
Moreover, by definition of the lexicographic order, they share the same rank value $r^*$.
If each partial order $(\mathcal{U}, \preceq_S)$ for $S \in \mathcal{L}$ exposes a single element, $\preceq_{\mathcal{B}}$ will also expose a single element.

We also adapt the construction of $x^*_{e^*}$ for bottleneck elements.
This time, we consider all rank differences $r(W,S)^+ - r(W,S')^+$ of sets $W \in \mathcal{U}$ and set $x^*_{e^*}$ sufficiently large as to cover \emph{all} these differences.
The element $S' = \phi_{e^*}(S) \in \mathcal{L}$ was chosen in such a way that it is the element $S'$ that is considered in the subsequent iteration of the main loop.
This will ensure primal feasibility.

Finally, we add an additional cleanup phase.
This will be beneficial, as variables from later iterations may render variables from previous iterations redundant.
In this case, we may carefully decrease variables in a post-processing step.
In general, deciding if a variable can be decreased by one may be a non-trivial task.
Moreover, determining the maximum in Line~\ref{alg:product:set-x}) is not simple, either.
In Table \ref{table:intro:result-summary} we provided some examples in which this is possible.

\sectionheadline{Feasibility}

Similar to Section \ref{sec:simple-version}, we can show that the truncation of a greedy product system does not cut off any integer feasible points.
Moreover, we can show that the greedy algorithm always obtains feasible primal solutions.

\begin{restatable}{lemma}{lemPlTruncationFeasible}
 Let $(A,r)$ be a greedy product system with truncation $(A',r)$ and let $x \in \ZZ_+^{|E|}$.
 Then $x \in (\ref{LP:P:product-full})$ if and only if $x \in (\ref{LP:T:product-full})$.
 \label{lem:pl:truncation-feasible}
\end{restatable}

\begin{proof}
 We will apply Lemma \ref{lem:sv:truncation-feasible} to each subsystem $(A,r)_{|U}$ in order to prove the lemma.
 For $U \in \mathcal{U}$, let us consider the polyhedron described by system $(A,r)_{|U}$, that is,
 \begin{align*}
  P_U = \left\{ x \in \ZZ_+^{|E|} : a_{(U,S)} x \geq r(U,S) \; \forall S \in \mathcal{L} \right\}.
 \end{align*}
 For each $U \in \mathcal{U}$, the set of integer points in $P_U$ is equal to the set of integer points in
 \begin{align*}
  P'_U = \left\{ x \in \ZZ_+^{|E|} : a'_{(U,S)} x \geq r(U,S) \; \forall S \in \mathcal{L} \right\}
 \end{align*}
 by application of Lemma \ref{lem:sv:truncation-feasible}.
 Here, we explicitly use the fact that $(A,r)$ is a greedy product system, hence $(A,r)_{|U}$ is a product system and the lemma is applicable.
 Since a polyhedron described by a matrix can be seen as the intersection of the polyhedra described by any partition of the rows of the matrix, the following holds:
 \begin{align*}
  \{ x \in \ZZ_+^{|E|} : A x \geq r \}
  = \bigcap_{U \in \mathcal{U}} P_U
  = \bigcap_{U \in \mathcal{U}} P'_U
  = \{ x \in \ZZ_+^{|E|} : A' x \geq r \}.
 \end{align*}
 The proof is concluded.
\end{proof}

\begin{restatable}{lemma}{lemPlGreedyTruncationFeasible}
 The greedy algorithm applied to the truncation (\ref{LP:T}) of a greedy system $(A,r)$ obtains a feasible primal solution to (\ref{LP:T:product-full}) and (\ref{LP:P:product-full}).
 \label{lem:pl:greedy-truncation-feasible}
\end{restatable}

\begin{proof}
 Let $x^* \in \ZZ_+^{|E|}$ be the solution returned by the algorithm without execution of Line~\ref{alg:product:cleanup}.
 Let $S_{\ell+1} \prec S_\ell \prec \dots \prec S_1$ be the sets $S_i$ considered by the algorithm and let $e_1,\dots,e_\ell$ be the bottleneck elements.
 Then $S_{i+1} = \phi_{e_i}(S_i)$ is satisfied.

 Clearly, the cleanup phase will not render $x^*$ infeasible, if it was feasible before.
 We will show that $x^*$ is feasible for each polyhedron induced by subsystem $(A,r)_{|U}, U \in \mathcal{U}$ individually.
 The proof will be concluded.
 Let $U \in \mathcal{U}$ and consider the polyhedron
 \begin{align*}
  P_U = \left\{ x \in \ZZ_+^{|E|} : a_{(U,S)} x \geq r(U,S) \; \forall S \in \mathcal{L} \right\}.
 \end{align*}
 Note that Line~\ref{appendix:alg:product:set-x}) ensures
 \begin{align*}
  x^*_{e_i} \geq \left\lceil \frac{r(U,S_i)^+ - r(U,S_{i+1})^+}{a'_{(U,S_i),e_i}} \right\rceil \geq \left\lceil \frac{r(U,S_i)^+ - r(U,S_{i+1})^+}{a_{(U,S_i),e_i}} \right\rceil.
 \end{align*}
 Moreover, $r(U,S_{\ell+1}) \leq 0$.
 The proof can be concluded analogously to the proof of Lemma \ref{lem:sv:greedy-truncation-feasible}:
 there, we used the same proof as in Lemma \ref{lem:sv:set-feasible} in order to obtain feasibility in $P_U$.
 Since the polyhedron considered in (\ref{LP:P:product-full}) is the intersection $\cap_{U \in \mathcal{U}} P_U$, feasibility in (\ref{LP:P:product-full}) is implied.
 Feasibility in (\ref{LP:T:product-full}) is concluded by application of Lemma \ref{lem:pl:truncation-feasible}.
\end{proof}

\sectionheadline{Approximation guarantee for the revised greedy algorithm}

In contrast to Section \ref{sec:simple-version}, the cleanup phase is necessary in order to obtain good approximation guarantees for greedy product systems.

\lemPlCleanupNecessary*

\begin{proof}
Let us consider the graph $G = (V,E)$ with vertices $V = \{v_0,\dots,v_n\}$ with $s = v_0$ and $t = v_n$ and edges $E = \left\{ \{s,v_i\} : 1 \leq i \leq n \right\}$.
The edge set $E$ will also be the set of columns in the greedy product system that we will consider.

The partially ordered set $\mathcal{U}$ will be the s-t-cut-lattice consisting of all outgoing edges of cuts in $G$.
That is, $\mathcal{U} = \{ \delta^+(V') : V' \subset V, s \in V', t \not\in V' \}$ with the natural ordering (for all $S \in \mathcal{L}$) for $\delta^+(V_1),\delta^+(V_2) \in \mathcal{U}$, $\delta^+(V_1) \preceq_S \delta^+(V_2)$ if and only if $V_2 \subseteq V_1$.
A visualization can be found in Figure~\ref{fig:pl:cleanup-necessary}.
The lattice $\mathcal{L} = 2^E$ will be the Boolean lattice.

We define the rank function $r(U,S) = 1 - |U \setminus S|$ and $a_{(U,S),e} = 1$, if $e \in U \cap S$ and $a_{(U,S),e} = 0$, otherwise and set $c_{\{s,v_i\}} = i$.
The truncation will have the coefficients $a'_{(U,S),e} = 1$, if $e \in U \cap S$ and $|U \setminus S| = 0$, or zero otherwise.
Let $U_i = \{v_0,\dots,v_i\}$ and $S_i = E \setminus \{ \{s,v_j\} : 1 \leq i \leq j \}$ with $S_0 = E$.

The revised greedy algorithm will set $y^*_{(U_i, S_i)} = 1$ for all $0 \leq i < n$.
Moreover, it will set $x^*_e = 1$ for all $e \in E$.
If we consider the constraint for $(U_0, S_0)$, we have $1 = r(U_0,S_0) \leq \sum_{e \in E} x^*_e = n$.
Hence, if we try to bound the solution cost as in Theorem \ref{thm:sv:greedy-approximation}, the approximation factor would be at least of order $|E|$.

However, the cleanup phase would set $x^*_e = 0$ for all $e \neq \{s,v_n\}$, and an analysis similar to the one in Theorem \ref{thm:sv:greedy-approximation} would yield optimality of this solution.
 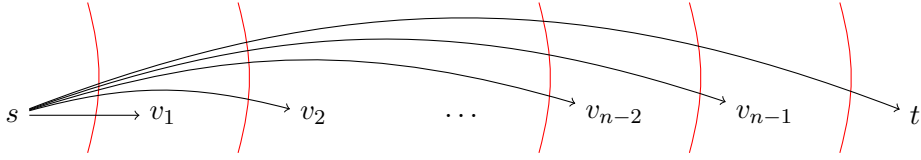
\begin{figure}[tb]
 \centering
 \begin{tikzpicture}
  \node[] (s) at (0,0) {$s$};
  \node[] (v1) at (2,0) {$v_1$};
  \node[] (v2) at (4,0) {$v_2$};
  \node[] (dots) at (6,0) {$\dots$};
  \node[] (v3) at (8,0) {$v_{n-2}$};
  \node[] (v4) at (10,0) {$v_{n-1}$};
  \node[] (t) at (12,0) {$t$};

  \draw[-,bend right=15,draw=red] (1,-0.5) to (1,1.5);
  \draw[-,bend right=15,draw=red] (3,-0.5) to (3,1.5);
  \draw[-,bend right=15,draw=red] (7,-0.5) to (7,1.5);
  \draw[-,bend right=15,draw=red] (9,-0.5) to (9,1.5);
  \draw[-,bend right=15,draw=red] (11,-0.5) to (11,1.5);

  \draw[->,bend left=0] (s) to (v1);
  \draw[->,bend left=15] (s) to (v2);
  \draw[->,bend left=17] (s) to (v3);
  \draw[->,bend left=19] (s) to (v4);
  \draw[->,bend left=21] (s) to (t);
 \end{tikzpicture}
 \caption{An example with bad optimality gap without cleanup. The sets $U_i$ of the dual greedy solution are marked in red.}
 \label{fig:pl:cleanup-necessary}
\end{figure}
\end{proof}

In order to characterize the influence on elements in terms of the cleanup phase, let us consider a solution $x^* \in \ZZ_+^{|E|}$ obtained by the revised greedy algorithm.
To get an intuition, let us assume for a second that the algorithm increased a single dual variable in each iteration, that is, Line~\ref{alg:product:select-rows}) returned a single maximum tuple in each iteration.
Let $(U_{\ell+1},S_{\ell+1}) \prec \dots \prec (U_1, S_1)$ be the constructed dual chain $e_i \in S_i \setminus S_{i+1}, 1 \leq i \leq \ell$ be the bottleneck elements.
As in Section~\ref{sec:simple-version}, $r(U_{\ell+1},S_{\ell+1}) \leq 0 < r(U_{\ell},S_{\ell})$.

During the cleanup phase, the value $x^*_{e_i}$ of element $e_i$ was not further reduced because either $x^*_{e_i} = 0$, or there is at least one tuple $(U,S) \in \mathcal{B}$ such that
$$\sum_{f \in U \cap S} a'_{(U,S),f} x^*_f - a'_{(U,S),e_i} < r(U,S) \leq \sum_{f \in U \cap S} a'_{(U,S),f} x_f.$$
We call this tuple $(U,S)$ a \emph{witness} of bottleneck element $e_i$.
Note that this tuple was not necessarily considered in Line~\ref{alg:product:select-rows}).

But let us suppose that some element $e_i$ has a witness $(U_t, S_t)$ on the dual chain.
Then $i \geq t$, otherwise $e_i \not\in S_t$.
The definition of witnesses implies
$$\sum_{f \in U_t \cap S_t} a'_{(U_t,S_t),f} x^*_f \leq 2 r(U_t,S_t).$$
In other words, if \emph{every} tuple $(U_t,S_t)$ of the dual chain was a witness for some element $e_i$, then $x^*$ would be a $2$-approximation for (\ref{LP:T}) by standard primal-dual approximation arguments (c.f.\ proof of Theorem~\ref{thm:sv:greedy-approximation}).
Of course, we can not expect this to happen in general.
But the following observation establishes a strong connection between witnesses and elements on the dual chain.
We will now cover the case that (possibly) multiple dual variables were increased simultaneously.

We define a \emph{multiplicity witness-cover} as follows.
Let $\mathcal{I} \subseteq \mathcal{B}$ be a family of tuples that were increased simultaneously in one iteration of the revised algorithm.
In this case, the rank value $r^*$ of all these tuples equals by definition of $(\mathcal{B}, \preceq_{\mathcal{B}})$.

We call $\mathcal{C} \subseteq \mathcal{B}$ a \emph{multiplicity witness-cover} of $\mathcal{I}$, if each tuple $(U,S) \in \mathcal{C}$ is a witness for some element $e \in E$, $r(U,S) \leq r^*$ and
every element $e \in E$ with $x^*_e > 0$ appears at least as often in $\mathcal{C}$, as it appears in $\mathcal{I}$.
That is, for all $e \in E$,
$$|\{ (U,S) \in \mathcal{I} : a'_{(U,S),e} x^*_e > 0 \}| \leq |\{ (U,S) \in \mathcal{C} : a'_{(U,S),e} x^*_e > 0 \}|.$$
If $\mathcal{C}$ is of small cardinality, we can show that $x^*$ is a good approximation.

We generalize our definition of $\delta$ from Section \ref{sec:simple-version} slightly to cover this case.
Given $(U,\emptyset), (W,S) \in \mathcal{B}, e \in E$, let
\begin{align*}
& \delta_{U,(W,S),e} = \begin{cases} \frac{a'_{(U,\emptyset),e}}{a'_{(W,S),e}} & a'_{(W,S),e} > 0 \text{ and } (r(W,\phi_e(S)) \geq 0 \text{ or } a'_{(W,S),e} = a_{(W,S),e}), \\ 1 & \text{otherwise}, \end{cases}
\end{align*}
and set $\delta = \max_{U,(W,S),e} \left\{ \delta_{U,(W,S),e} \right\}$.
The following Theorem \ref{thm:ps:greedy-approximation} yields bounds on the solution cost, depending on $A$ being binary or a general matrix.
Finally, Proposition \ref{prp:ps:greedy-bad-dual-solution} shows that the dependency on $k$ is inherent in the type of dual solution constructed.

\thmPsMultiDualGreedyApproximation*

\begin{proof}
 In order to prove the result, we will first reformulate the total cost of the solution in terms of the dual variables.
 Therefore, let $\epsilon_i$ be the value by which the dual variables in iteration $i$ were increased and let $r_i$ be the rank value of all sets $(U,S) \in \mathcal{I}_i$.
 We will show the following:
 \begin{align*}
   \sum_{e \in E} c_e x^*_e &= \sum_{e \in E} x^*_e \sum_{(U,S) \in \mathcal{L}} a'_{(U,S),e} y^*_{(U,S)}
   = \sum_{(U,S) \in \mathcal{L}} y_{(U,S)} \sum_{e \in U \cap S} a'_{(U,S),e} x^*_e \\
   &= \sum_{i=1}^{\ell} \epsilon_i \sum_{(U,S) \in \mathcal{I}_i} \sum_{e \in U \cap S} a'_{(U,S),e} x^*_e
   \leq (\delta + 1) \sum_{i=1}^{\ell} \epsilon_i |\mathcal{C}_i| r_i \\
   &\leq k (\delta + 1) \sum_{i=1}^{\ell} \epsilon_i r_i |\mathcal{I}_i|
   = k (\delta + 1) \sum_{(U,S) \in \mathcal{L}} y^*_{(U,S)} r(U,S)
   \leq k (\delta + 1) \text{OPT}.
  \end{align*}
  Except for the first inequality, the remaining steps are simple calculus.
  Hence, the proof is concluded if we can show that for every iteration $i$, the following holds:
  $$\sum_{(U,S) \in \mathcal{I}_i} \sum_{e \in U \cap S} a'_{(U,S),e} x^*_e \leq (\delta + 1) |\mathcal{C}_i| r_i.$$
  For every $(U,S) \in \mathcal{C}_i$, let $e_{(U,S)} \in U \cap S$ be a witness for $(U,S)$, that is, with $a'_{(U,S),e_{(U,S)}} x^*_{e_{(U,S)}} > 0$ and $$\sum_{f \in U \cap S} a'_{(U,S),f} x_f - a'_{(U,S),e_{(U,S)}} < r(U,S).$$
  By definition of the truncation, this implies that $a'_{(U,S),e_{(U,S)}} \leq r_i$. Then
  \begin{align*}
   \sum_{(U,S) \in \mathcal{I}_i} \sum_{e \in U \cap S} a'_{(U,S),e} x^*_e
   &\leq \sum_{(W,T) \in \mathcal{C}_i} \sum_{e \in W \cap T} a'_{(U,S),e} x^*_e \\
   &= \sum_{(W,T) \in \mathcal{C}_i} \left( \sum_{e \in W \cap T} a'_{(U,S),e} x^*_e - a'_{(U,S),e_{(W,T)}} + a'_{(U,S),e_{(W,T)}} \right) \\
   &\leq \sum_{(W,T) \in \mathcal{C}_i} \left( \delta \left( \sum_{f \in W \cap T} a'_{(W,T),f} x^*_f - a'_{(W,T),e_{(W,T)}} \right) + r_i \right) \\
   &< \sum_{(W,T) \in \mathcal{C}_i} \left( \delta r(W,T) + r_i \right)
   \leq (\delta + 1) |\mathcal{C}_i| r_i.
  \end{align*}
  The final step is due to the monotonicity of $r$ and the choice of the elements to be increased by the algorithm.

  If $A'$ is binary, then the previous step can be modified in order to reflect
  \begin{align*}
   & \sum_{(U,S) \in \mathcal{I}_i} \sum_{e \in U \cap S} a'_{(U,S),e} x^*_e
   \leq \sum_{(W,T) \in \mathcal{C}_i} \sum_{e \in W \cap T} x^*_e
   \leq \sum_{(W,T) \in \mathcal{C}_i} r(W,T)
   \leq |\mathcal{C}_i| r_i. \qedhere
  \end{align*}
\end{proof}

\prpPsGreedyLowerBound*

\begin{proof}
 For $k \in \ZZ_+$, let $E = \{1,\dots,k\}$ and let $\mathcal{U} = \left\{ E \right\} \cup \left\{ \{i\} : 1 \leq i \leq k  \right\}$, $\mathcal{L} = 2^E$ is the Boolean lattice.
 We assume that for all $S \in \mathcal{L}$, the partial order $(\mathcal{U}, \preceq_S)$ is very simple: $U \preceq U'$ if and only if $U' = \{E\}$.
 That is, the partial order exposes the set $\{E\} \in \mathcal{U}$ instead of a singleton.
 Moreover, let us define the rank function $r(U,S) = 1 - | U \setminus S |$.
 That is, $r$ is positive only if $U \setminus S = \emptyset$.
 Let $c_e = 1$ for all $e \in E$.

 The greedy algorithm will set $y^*_{(E,E)} = 1$.
 At this point, all elements become tight.
 Moreover, the greedy algorithm will set $x^*_e = 1$ for all $e \in E$ in order to become feasible.
 The only witness cover for $(E,E)$ is the family $\mathcal{C} = \left\{ (\{i\},\emptyset) : 1 \leq i \leq k \right\}$.

 An optimum dual solution would set $y^*_{(\{i\},E)} = 1$ for all $1 \leq i \leq k$ with objective value $k$, matching the primal solution cost.
 This proves that the instance does not have any gap between an optimum primal and an optimum dual solution.
\end{proof}

\subsection*{Applications}
Note, in many proofs in this section, it is convenient to define the matrix $A$ in terms of the complement.
That is, we consider matrices with respect to a lattice $\mathcal{L}$ with support $supp(a_S) = E \setminus S$ of row $S \in \mathcal{L}$.
Note that the same results hold in this case.
Sometimes it is more convenient to formulate problems of this type.

\begin{restatable}{lemma}{lemAppDelta}
 For a greedy system $(A,r)$, let us define
 \begin{align*}
 \beta = \max_{S,e} \left\{ \frac{a_{E,e}}{a_{S,e}} : a_{S,e} > 0 \right\} \quad \text{and} \quad \gamma = \max_{S,e} \left\{ \frac{a_{S,e}}{r(S) - r(\phi_e(S))} \right\}.
 \end{align*}

 Then $\delta \leq \beta \gamma$.
 \label{lem:app:delta}
\end{restatable}

\begin{proof}
 For $S \in \mathcal{L}, e \in S, S' = \phi_e(S)$ with $r(S') \geq 0$ and $a'_{S,e} > 0$ and $T = \phi_e(E)$, we have:
 $$a_{S,e} \geq \frac{1}{\beta} a_{E,e} \geq \frac{1}{\beta \delta} \left( r(E) - r(T) \right).$$
 Moreover, we have:
 $$r(S) - r(S') \geq \frac{1}{\delta} a_{S,e} \geq \frac{1}{\beta \delta} a_{E,e}.$$
 Hence, if $a'_{S,e} = a_{S,e}$, we get
 $$ \delta_{S,e} \leq \beta \delta \frac{a'_{E,e}}{r(E) - r(T)} \leq \beta \delta \frac{r(E) - r(T)}{r(E) - r(T)} = \beta \delta$$
 and if $a'_{S,e} = r(S) - r(S')$, we get
 $$ \delta_{S,e} \leq \beta \delta \frac{a_{E,e}}{a_{E,e}} = \beta \delta,$$
 which concludes the proof.
\end{proof}

\begin{restatable}{lemma}{lemAppB}
 If $a_{S,e} \geq r(S) - r(S')$ holds for all $S \in \mathcal{L}, e \in S, S' = \phi_e(S)$ with $r(S') \geq 0$, then $b = 1$ in Theorem \ref{thm:sv:greedy-approximation}.
 \label{lem:app:b}
\end{restatable}

\begin{proof}
 $a_{S,e} \geq r(S) - r(S')$ implies that $a'_{S,e} = r(S)^+ - r(S')^+$ holds for all such sets, which implies that $b = 1$.
\end{proof}

\begin{restatable}{lemma}{lemAppOnALine}
 Let $(A,r)$ be a greedy product system with groundset $E$ and incomparable elements $\mathcal{T} = \{T_1,\dots,T_\ell\}$.
 Suppose there are mappings $a,b:E \rightarrow \{1,\dots,\ell\}$ with $a(e) \leq t \leq b(e)$ if and only if $e \in T_t$.
 Then for every element $(T,S) \in \mathcal{L}$ there is a witness cover of size at most $2$.
 \label{lem:app:on-a-line}
\end{restatable}

\begin{proof}
 If $(T,S)$ was a witness, it would also be a witness cover for itself.
 Hence, let us assume it was not a witness cover.
 Let $T = T_t$ and let $a,b$ be mappings such that $e \in T$ if and only if $a(e) \leq t \leq b(e)$.
 With respect to $a$ and $b$, let $t_1 < t < t_2$ be witnesses such that $t_1$ is chosen maximal, $t_2$ is chosen minimal with this property.
 The case that only one of the two exists will be handled later.

 Now, suppose $\left\{ (T_{t_1},S), (T_{t_2},S) \right\}$ was not a witness cover for $(T,S)$.
 Then there is some element $e \in T \setminus S \setminus T_{t_1} \setminus T_{t_2}$ with $x_e > 0$.
 The mapping functions $a$ and $b$ imply that $t_1 < a(e) \leq t \leq b(e) < t_2$.

 But then, either there exists a witness $a(e) \leq t' \leq b(e)$, which contradicts to the choice of $t_1$ or $t_2$.
 Or, if there is no such witness, it contradicts to the cleanup phase.
 In this case, $x_e$ could be reduced by at least one.

 Finally, let us assume that there was just one witness, say $t_1 < t$, but no witness $t < t_2$.
 In this case, the same arguments apply.
 If there is some element with $t_1 < a(e) \leq t \leq b(e)$, either there is another witness $t_2 > t$ or a larger witness $t_1 < t' \leq t$.
 By symmetry of this argument, the proof is concluded.
\end{proof}

\begin{restatable}{corollary}{corAppOnMultipleLine}
 Suppose there is a parameter $k$ and mappings $a^i,b^i:E \rightarrow S_i, 1 \leq i \leq k$, where $S_i = \{s^i_1,\dots,s^i_{|S_i|}\} \subseteq \{1,\dots,\ell\}$.
 Let $(A,r)$ be a greedy product system with groundset $E$ and incomparable elements $\mathcal{T} = \{T_1,\dots,T_\ell\}$.
 Moreover, let $\pi: E \rightarrow \{1,\dots,k\}$ be an assignment of elements to mappings.

 Suppose that for every $e \in E$ and every $1 \leq j \leq |S_{\pi(e)}|$, $e \in T_{s_j}$ if and only if $a^{\pi(e)}(e) \leq j \leq b^{\pi(e)}(e)$.
 Then for every element $(T,S) \in \mathcal{L}$ there is a witness cover of size at most $2k$.
 \label{cor:app:on-multiple-lines}
\end{restatable}

\begin{proof}
 Let $(W,T) \in \mathcal{L}$ and suppose it is not a witness.
 For each mapping $i$, let $Q_i$ be the set of elements $e \in Q_i \subseteq W \setminus T$ with $x_e > 0$ and $\pi(e) = i$.

 Then, for every index $i$, Lemma \ref{lem:app:on-a-line} obtains a witness-cover for all elements in $Q_i$ of size at most $2$.
 The union of covers for each index intersecting $e$ yields a witness-cover of size at most $2k$.
\end{proof}

\begin{restatable}{application}{appContraPolymatroids}[Optimization over contra-polymatroids]
 Given a supermodular, non-negative rank $r$ and elements $E$ with cost $c_e$.
 Find a solution to $\min\{ c^T x : \sum_{e \in S} x_e \geq r(S)\; \forall S \subseteq E \}$.

 $\rightarrow$ Theorem \ref{thm:sv:greedy-approximation} provides an optimum solution, which coincides with \cite{edmonds1970submodular}.
\end{restatable}

\begin{proof}
 The problem naturally fits into the form (\ref{LP:P}) with the Boolean lattice $\mathcal{L} = 2^E$ and $a_{S,e} = 1$, if $e \in S$, and zero, otherwise.
 In this case, $\delta = 1$, $b = 1$ and $a = 0$.
 Moreover, the rank is non-negative.
\end{proof}

\begin{restatable}{application}{appKnapsackCover}[Knapsack cover]
 Given Elements $E$ with weight $u_e$, cost $c_e$ and multiplicity $d_e$, find a multiset of elements, each with multiplicity at most $d_e$, with minimum cost covering $D$.

 $\rightarrow$ Theorem \ref{thm:sv:greedy-approximation} shows a 2-approximation which coincides with \cite{carnes2008primal} (or \cite{mccormick2016primal} for $d \not\equiv 1$).
 Also holds if $c$ is a separable, monotone increasing, convex function and $u$ is a separable, monotone decreasing, concave function.
 In this case, the running time depends linearly on $\max_e d_e$.
\end{restatable}

\begin{proof}
 For this proof, we consider the complement of supports in matrix $A$.
 Let $|E| = n$.
 First, let us consider the case of $d \equiv 1$.
 Define $A$ on the Boolean lattice with $a_{S,e} = \begin{cases} u_e & e \not\in S \\ 0 & e \in S\end{cases}$.
 Moreover, let $r(S) = D - \sum_{e \in E} u_e$.
 Then $a = 1$ (in Theorem \ref{thm:sv:greedy-approximation}) and $b = 1$.
 The latter is due to the linear relation between $r$ and $A$, we have $r(S) - r(S \cup \{e\}) = u_e$.
 For the same reason, $\delta = 1$, as $a'_{S,e} < u_e$, only if $r(S \cup \{e\}) < 0$.

 Now, for item multiplicities, we can not ensure that $b = 1$, as for $r(S') \leq 0$, the truncation may yield that the algorithm still has to perform rounding for the final element chosen.
 This would only get a 3-approximation with Theorem \ref{thm:sv:greedy-approximation}.

 Instead, let us consider the ring-family $\mathcal{L}$ induced by the ideals of the following partial order.
 Let $e_i^j, 1 \leq i \leq n, 1 \leq j \leq d_i$ with $e_i^j \prec e_i^{j+1}$ for all $i,j$.
 That is, the copies of each element form a chain.
 Set $c_{e_i^j} = c_i$ and $u_{e_i^j} = u_i$ for all elements, that is, each element gets the unit cost and unit increase from the knapsack elements.
 We also use the same rank function as before, that is, $r(S) = D - \sum_{e \in S} u_e$ for $S \in \mathcal{L}$.
 The instance has the same parameters as before, hence, yields a 2-approximation with Theorem \ref{thm:sv:greedy-approximation}.

 Unfortunately, the instance has pseudopolynomially many elements.
 But a careful implementation of the algorithm can handle this implicitly.
 There is a bijection between vectors $v \in \ZZ_+^E$ with $0 \leq v_e \leq d_e$ and the ideals in $\mathcal{L}$.

 Since all copies of an element $e_i^j$ are equivalent, all copies will become tight at the same time.
 Let us assume that the copies of element $e$ became tight.
 At this point, instead of adding one copy after another in up to $d_e$ iterations, we may compute $s = \min\{ \frac{r(S)}{u_e}, d_e \}$.
 If $d_e = s$, add all copies to the solution and iterate.
 Otherwise, add only the first $s$ copies of element $e$ to the solution.
 The algorithm terminates by choice of $s$ and $r$.

 Finally, if $c$ and $u$ are separable, monotone and convex and concave, respectively, we can use the same construction as above.
 We set the weight and cost of each copy to the marginal differences of the functions, that is, $c_{e_i^j} = c^i(j) - c^i(j-1)$ and $u_{e_i^j} = u^i(j) - u^i(j-1)$.

 From every solution $x$ to (\ref{LP:P}), we can construct a solution to the knapsack problem.
 If the support of solution $x$ was not an ideal in $\mathcal{L}$, we can replace it to become one in a post-processing step without increasing the cost (as $c$ was supposed to be convex) and without decreasing the weight (as $u$ was supposed to be concave).
 Moreover, there is a bijection from feasible solution vectors of the knapsack problem to feasible solutions in (\ref{LP:P}), which are ideals in $\mathcal{L}$.
 In this case, copies of an element no longer become tight at the same time, however, due to $c$ and $u$, they will become tight in the order of an ideal.
 Hence, the algorithm will find solutions which are an ideal.
 The running time will depend linearly on $\max_e d_e$.
\end{proof}

\begin{restatable}{application}{appSubsetCover}[Subset cover]
 Given a groundset $G$ and subsets $T_i \subseteq G, i \in I$ with cost $c_i$, find a set $S \subseteq I$ with $\cup_{i \in S} T_i = G$ of minimum cost.

 $\rightarrow$ Theorem \ref{thm:sv:greedy-approximation} provides a $\max_i |T_i|$-approximation.
 The best known greedy approximation for this problem yields a $\log(\max_i |T_i|)$-approximation \cite{chvatal1979greedy}.
\end{restatable}

\begin{proof}
 In the construction from Proposition \ref{prp:sv:log-inapproximable}, we have $\beta = 1$ and $\delta = \max_i |T_i|$ in Lemma \ref{lem:app:delta}.
 Hence, $\delta \leq \beta \gamma = \max_i |T_i|$ in Theorem \ref{thm:sv:greedy-approximation}.
 Moreover, $a=0$, as $r$ is non-negative and $b=1$ due to Lemma~\ref{lem:app:b}.
\end{proof}

\begin{restatable}{application}{appMatroidIntersection}[Contra-polymatroid intersection]
 Given $p$ contra-polymatroids $(E,\mathcal{F}_i), 1 \leq i \leq p$, find a vector $x \in \mathcal{F}_1 \cap \dots \cap \mathcal{F}_p$ of minimum cost $c^T x$.

 $\rightarrow$ Theorem \ref{thm:ps:greedy-approximation} yields a $p$-approximation, matching \cite{jenkyns1976efficacy}.
\end{restatable}

\begin{proof}
 The value $k$ in Theorem \ref{thm:ps:greedy-approximation} is always bounded by $|\mathcal{U}|$.
\end{proof}

\begin{restatable}{application}{appFlowCoverOnMultiplePath}[Flow cover on $k$ lines]
 Let $G = (V,F)$ be a graph and let $P_1,\dots,P_k$ be paths in $G$ such that $\bigcup_i P_i = F$.
 Moreover, let $E_i$ be a family of subpaths of $P_i$.
 Let $E = \bigcup_i E_i$ and define a weight $u: E \rightarrow \ZZ_+$, cost $c: E \rightarrow \ZZ_+$ and let $D: F \rightarrow \ZZ_+$ be a demand function.
 Find a subset of $E$ such that every demand is covered and the solution has minimum cost.

 $\rightarrow$ Theorem \ref{thm:ps:greedy-approximation} provides a $4k$-approximation, where $k = \max_{e \in E} |\{i : e \in P_i \}|$ is the maximum number of paths an edge appears in.
 For $k=1$, this matches the best known approximation, obtained by \cite{bar2001unified}. To the best of our knowledge, no result for $k > 1$ was known.
 \label{app:flow-cover-on-k-paths}
\end{restatable}

\begin{proof}
 The problem can be formulated as
 $$\min \left\{ \sum_{P \in E} c_P x_P : \sum_{P : e \in P} u_P x_P \geq r(U,S) \; \forall (U,S) \in \mathcal{L} \right\},$$
 with $r(U,S) = D(U) - \sum_{P \in U \cap S} u_P$.
 Apply Corollary \ref{cor:app:on-multiple-lines}.

 Therefore, for every subpath $Q$ of $P_i$, define $\pi(Q) = i$.
 Moreover, for every path, let $S_i = P_i$ be the ordered set of edges in $P_i$.
 Moreover, define $a^{\pi(Q)}(Q)$ and $b^{\pi(Q)}(Q)$ as the distance on path $P_i$ to the startpoint, respectively, the endpoint of $Q$.
\end{proof}

\begin{restatable}{application}{appKnapsackCoverPrecedenceConstraints}[Knapsack cover with precedence constraints]
Same input as knapsack cover but, additionally, a directed, acyclic graph $G = (E,A)$ is given.
Find an ideal in $G$ with weight at least $D$ of minimum cost.

 $\rightarrow$ Theorem \ref{thm:ps:greedy-approximation} provides a $w$-approximation, where $w$ is the width of $G$ (matching \cite{mccormick2016primal}).
 The same bound holds if there are $m$ demands to be covered.
 \label{app:knapsack-cover-precedence-constraints}
\end{restatable}

\begin{proof}
 McCormick et al. have shown that knapsack cover with precedence constraints can be formulated in the following way \cite{mccormick2016primal}.
 For an ideal $A$ in $G$, let $\mathcal{P}(A) = \{ e \in E : |\delta^+(v) \setminus A| = 0 \}$ be the set of vertices without any incoming edges, given that all elements from $A$ were deleted.
 Now, let us define $\mathcal{T} = \{ \mathcal{P}(A) : \sum_{e \in A} u_e < D \}$ as the family of all sets of vertices without incoming edges obtained after removal of an ideal with weight below $D$.
 We use the ordering $\mathcal{P}(A) \preceq \mathcal{P}(A')$ if and only if $A \subseteq A'$.

 Moreover, let us define the rank function $r(U,S) = 1 - |U \cap S|$.
 Let $\mathcal{B} = 2^E$ be the Boolean lattice, then
 $$ \min \left\{ c^Tx : \sum_{e \in U \setminus S} x_e \geq r(U,S) \; \forall (U,S) \in \mathcal{L} \right\} $$
 is the formulation derived in \cite{mccormick2016primal}.

 If $w$ is the width of $G$, then every constraint contains at most $w$ elements.
 Hence, the maximum size of a witness cover is bounded by $w$.
 Note that $A'$ will be a binary matrix, hence, Theorem \ref{thm:ps:greedy-approximation} obtains the desired result.
\end{proof}

\begin{restatable}{application}{appGeneralizedSteinerTree}[Generalized steiner tree]
 Given a graph $G = (V,E)$ with cost function $c: E \rightarrow \ZZ_+$ and $k$ pairs of vertices $(s_i,t_i)$, find a subset of edges that contains an $s_i$-$t_i$ path for every $i$ and is of minimum cost.

 $\rightarrow$ Theorem \ref{thm:ps:greedy-approximation} provides a $2$-approximation, matching \cite{goemans1997primal}.
 \label{app:generalized-steiner-tree}
\end{restatable}

\begin{proof}
 We formulate the problem as discussed in \cite{goemans1997primal}, that is, as the following LP relaxation:
 $$ \min \left\{ \sum_{e \in E} c_e x_e : \sum_{e \in \delta(S)} x_e \geq 1 \; \forall S \subseteq V \text{ with } |\{s_i,t_i\} \cap S | = 1 \right\}.$$
 In our notation, the problem will be formulated with $\mathcal{B} = 2^E$ as the Boolean lattice.
 The set family $\mathcal{T} = \{ \delta(S) : S \subseteq V \text{ with } |\{s_i,t_i\} \cap S | = 1 \}$ and system $(A,r)$ with
 $$a_{(U,S),e} = \begin{cases} 1, & e \in U, |U \cap S| = 0 \\ 0, & \text{else} \end{cases} \quad\quad\quad\quad r(U,S) = 1 - |U \cap S|.$$
 Moreover, the ordering $\preceq_S$ will be the following partition $\mathcal{U}_1 \cup \mathcal{U}_2$ with $U \prec U'$ if and only if $U \in \mathcal{U}_1, U' \in \mathcal{U}_2$.
 Let $T_1,\dots,T_k$ be the maximal connected components in the graph $(V,S)$.
 Then $U \in \mathcal{U}_1$ if and only if there is a component $i$ with $\delta_G(T_i) = U$ and $\max_i |\{s_i,t_i\} \cap T_i| = 1$.

 In other words, the minimal elements in $\min (\mathcal{L} \setminus S)$ are the inclusionwise maximal components with respect to $S$ that contain at most one terminal from each terminal pair, and contain at least one terminal.
 For $S = \emptyset$, $\min \mathcal{L}$ consists of the family of vertex-sets containing single terminals.

 We want to apply Theorem \ref{thm:ps:greedy-approximation}.
 Hence, we need to show that for every iteration of the algorithm, there exists a small enough multiplicity witness-cover.

 Let us fix an iteration for $S \in \mathcal{B}$ and let $U_1,\dots,U_k$ be the maximal connected components in $(V,S)$.
 Let $G' = (V',F)$ be the graph obtained by contraction of the maximal connected components in $(V,S)$.
 Moreover, let $F = \{e \in E : x_e > 0\}$ be the set of edges selected by the algorithm connecting these components.

 By construction of the cleanup-phase, $F$ will be a forest in $G'$.
 If this was not the case, there is a cycle in $G'$, in other words, an edge could be removed in order to reduce the solution cost.

 We construct a multiplicity witness-cover as follows.
 Let $u_1,\dots,u_k$ be the vertices corresponding to the connected components in $(V,S)$.
 For every edge $e = \{u_i,v\}$ in $G'$, that corresponds to a set (U,S) increased by the algorithm in this iteration, we construct a witness $(\delta(U_{i,e}),S)$ as follows.
 Let $U_{i,e} \subseteq V$ be the set of vertices in the connected component of $u_i$ in $G'$ after removal of edge $e$.
 Then $|\delta(U_{i,e}) \cap F| = 1$, as $G'$ is a forest.
 Moreover, $\delta(U_{i,e})$ covers edge $e$ once.

 Any edge $e \in F$ will appear at most twice in $U_1,\dots,U_k$.
 It appears twice, if the edge connects two of the components $U_i$ and $U_j$.
 But in this case, the edge will appear in two witnesses, $(\delta(U_{i,e}),S)$ and $(\delta(U_{j,e}),S)$.
 Hence, the sets constructed are a multiplicity witness-cover.
 The number of sets constructed is at most $2k$.

 Hence, Theorem \ref{thm:ps:greedy-approximation} yields a $2$-approximation.
\end{proof}

\begin{restatable}{application}{appMinimumMulticutOnTrees}[Minimum multicut on trees]
 Given a tree $G = (V,E)$ with edge cost $c: E \rightarrow \ZZ_+$ and $k$ pairs of vertices $(s_i,t_i)$, find a minimum cost subset $E'$ of edges such that the forest $(V,E \setminus E')$ disconnects all $s_i$-$t_i$-pairs.

 $\rightarrow$ Theorem \ref{thm:ps:greedy-approximation} provides a $2$-approximation, matching \cite{garg1997primal}.
 \label{app:minimum-multicut-on-trees}
\end{restatable}

\begin{proof}
 The proof is essentially the same as in Theorem 4.10 in \cite{goemans1997primal}.
 We can formulate the problem with $\mathcal{B} = 2^E$ as the Boolean lattice.
 Let $P_i$ be the edge-set of the unique $s_i$-$t_i$-path in $G$, then $\mathcal{T} = \{P_1,\dots,P_k\}$ and
 $a_{(U,S),e} = 1$, if $e \in U \setminus S$, and zero otherwise.
 The rank is defined as $r(U,S) = 1 - |U \cap S|$.

 For this problem, we need a specific ordering of the sets in order to obtain the desired approximation guarantee.
 Therefore, let us root the tree at an arbitrary vertex $r$.
 We say that vertex $v$ lies at depth $d_v$, if the $r$-$v$-path in $G$ consists of $d_v$ edges.
 For a pair $(s_i,t_i)$, we define the least common ancestor as the vertex $x$ with smallest value $d_x$ which lies on the unique $s_i$-$t_i$ path, call this value $d_i$ for pair $i$.

 Then $P_i \preceq_S P_{i'}$ if and only if $d_i \geq d_{i'}$.
 The revised algorithm applied to the system $(A,r)$ is essentially the algorithm discussed in \cite{garg1997primal}.

 In order to prove the theorem, we need to show that for every $(U,S) \in \mathcal{L}$ increased by the algorithm, there is a witness-cover of size at most $2$.
 We will show the following.
 For every path $P_i$ there are at most two edges $e_1,e_2 \in P_i$ with positive value in the primal solution.
 The cleanup-phase implies that both edges have a witness $(U_j,S_j), j = 1,2$.
 Combination of both yields the desired witness-cover.
 For the sake of this proof, let $F = \{e \in E : x_e = 1\}$ be the support of the solution obtained by the revised greedy algorithm.

 For any given path $P_i$, for which the dual was increased, let $a_i$ be the least common ancestor.
 Let $T_1$ be the path from $s_i$ to $a_i$.
 We will see that $|T_1 \cap F| = 1$ (analogously for $T_2$, the path from $t_i$ to $a_i$).
 Suppose not, then $|T_1 \cap F| \geq 2$.
 Let us suppose that $F'$ is the solution obtained by removing all edges from $T_1 \cap F$ but the one nearest to $r$, call it $e^*$.
 We claim that $F'$ is still a feasible solution.

 Suppose not, let $T_j$ be any path for pair $j$ with $T_1 \cap T_j \neq \emptyset$.
 Then $T_1 \cap T_j$ is a subpath of $T_1$ from some vertex $v \in T_1$ to $a_i$.
 Otherwise, $T_j$ had a least common ancestor of depth $d_j > d_i$.
 This violates the ordering $\preceq_S$ in the iteration $i$ was chosen.
 Hence, if $F \cap T_j \cap T_1 \neq \emptyset$, then $e^* \in T_j$.
\end{proof}

\end{document}